\documentclass[aps,10pt,twocolumn,showpacs,prx,amsmath,amssymb,floatfix,superscriptaddress,dvipsnames]{revtex4-2}

\usepackage{amsthm}
\usepackage{mathtools}
\usepackage{physics}
\usepackage{bbm}
\usepackage{amsfonts}
\usepackage{amssymb}
\usepackage{graphicx}

\usepackage[T1]{fontenc}
\usepackage{bbold}
\usepackage{float}
\usepackage[utf8]{inputenc}
\usepackage{csquotes}

\usepackage[colorlinks=true,linktocpage=true]{hyperref}
\hypersetup{
    allcolors  = {blue},
}

\theoremstyle{plain}
\newtheorem{theorem}{Theorem}

\newtheorem{proposition}{Proposition}
\newtheorem{corollary}{Corollary}

\usepackage{array}
\usepackage{multirow}
\usepackage{xcolor}

\usepackage[utf8]{inputenc}
\usepackage[english]{babel}
\usepackage[T1]{fontenc}

\usepackage{amsfonts}
\usepackage{dsfont}

\usepackage{enumerate}
\usepackage{braket}
\usepackage{framed}
\usepackage{graphicx}

\usepackage{hyperref}
\usepackage{cleveref}
\usepackage{natbib}
\usepackage{csquotes}
\usepackage{xcolor}
\usepackage[normalem]{ulem}

\usepackage{bbold}

\definecolor{deeplilac}{rgb}{0.6, 0.33, 0.73}

\newcommand{\inl}{INL -- International Iberian Nanotechnology Laboratory, Av. Mestre Jos\'{e} Veiga s/n, 4715-330 Braga, Portugal} 

\newcommand{\pisa}{Department of Physics ``E. Fermi'', University of Pisa, Largo B. Pontecorvo 3, 56127 Pisa, Italy} 

\newcommand{\ifusp}{Department of Mathematical Physics, Institute of Physics, University of São Paulo,\\ Rua do Matão 1371, São Paulo 05508-090, São Paulo, Brazil} 

\newcommand{\cfum}{Centro de F\'{i}sica, Universidade do Minho, Campus de Gualtar, 4710-057 Braga, Portugal} 

\newcommand{\iip}{International Institute of Physics, Federal University of Rio Grande do Norte, 59078-970, Natal, Brazil} 

\newcommand{\ulm}{Institute of Theoretical Physics, Ulm University, Albert-Einstein-Allee 11 89081, Ulm, Germany}

\usepackage{soul}

\begin{document}

\title{Contextuality in anomalous heat flow}




\author{Naim E. Comar}
\email{naim.comar@usp.br}
\affiliation{\ifusp}
\affiliation{\iip}

\author{Danilo Cius}
\email{cius@if.usp.br}
\affiliation{\ifusp}

\author{Luis F. Santos}
\email{luisf@usp.br}
\affiliation{\ifusp}

\author{Rafael Wagner}
\email{rafael.wagner@inl.int}
\affiliation{\ulm}
\affiliation{\inl}
\affiliation{\cfum}
\affiliation{\pisa}
 
\author{Bárbara Amaral}
\email{bamaral@if.usp.br}
\affiliation{\ifusp}
 
\date{\today}

\begin{abstract}
In classical thermodynamics, heat must spontaneously flow from hot to cold systems. In quantum thermodynamics, the same law applies when considering multipartite product thermal states evolving unitarily. If initial correlations are present, anomalous heat flow can happen, temporarily making cold thermal states colder and hot thermal states hotter. Such an effect can happen due to entanglement, but also because of classical randomness, hence lacking a direct connection with nonclassicality. In this work, we introduce scenarios where anomalous heat flow \emph{does} have a direct link to nonclassicality, defined as the failure of noncontextual models to explain experimental data. We start by extending known noncontextuality inequalities to a setup where sequential transformations are considered. We then show a class of quantum prepare-transform-measure protocols, {characterized by a time interval $(0,\tau_c)$} for a given critical time $\tau_c$, where anomalous heat flow happens only if a noncontextuality inequality is violated. We also analyze a recent experiment from \href{https://doi.org/10.1038/s41467-019-10333-7}{Micadei \emph{et. al.} [Nat. Commun. 10, 2456 (2019)]} and find the critical time $\tau_c$ based on their experimental parameters. We conclude by investigating heat flow in the evolution of two {qudit} systems, showing that our findings are not an artifact of using two-qubit systems.
\end{abstract}

\maketitle
\section{Introduction}
The second law of thermodynamics forbids heat to flow from a colder system to a warmer system without consuming any resource, assuming the systems are isolated~\cite{callen2006thermodynamics,blundell2010concepts}. The direction of heat flow suggests the notion of an `arrow of time' \cite{eddington2018nature}. If one takes heat flow direction and the arrow of time as equivalent phenomena, given two initially correlated interacting systems, a local observer of one of these two systems could have their arrow of time deceived~\cite{Jennings_2010,vedral2016arrow}. For instance, the heat exchanged (in the absence of work) between these two systems can initially have a `backflow', i.e., heat can flow from the colder to the warmer system. Such heat flows are said to be \emph{anomalous}~\cite{lipkabartosik2024fundamental}. Crucially, this is \emph{not} a violation of Clausius' formulation of the second law of thermodynamics, since correlations are \emph{consumed} for such anomalies to happen. Their consumption is responsible for the possibility of a reversal in the heat flow.

An attempt to visualize what is happening is to imagine a `Maxwell demon'~\cite{leff1990maxwell1,leff2002maxwell2,Rex_Andrew_Rev} that has knowledge about the correlations between two physical systems and somehow wants to use this information to perform a thermodynamic task. The demon can use these correlations to make heat flow from the cold to the hot system. This thought experiment was first discussed by Lloyd~\cite{Lloyd}, and then concretely investigated for the case of pure quantum states by Partovi~\cite{Partovi2008}, and refined by others~\cite{Jennings_2010,Bera_2017,Jevtic_2012,lostaglio2020quasiprobability,Jevtic_2015,Sagawa_2012,Koski,CAI_2024129389,Henao_2018}. Recently, the prediction of heat flow anomaly in quantum theory has been tested~\cite{Micadei_2019}.

Interestingly, there is no no-go result forbidding the demon (possessing knowledge of initial correlations) to cause anomalous heat flow in microscopic models described by \emph{classical} statistical mechanics. Indeed, as shown in Ref.~\cite{Jevtic_2015}, one can construct examples of heat flow reversals caused by correlations without quantum coherence with respect to the local energy basis. Notwithstanding, there is a bound on the \emph{amount} of heat flow anomaly that can be caused by classical correlations. In such cases, it is possible to surpass this bound only with the presence of entanglement between the initial states~\cite{Jennings_2010}, this is known as the \textit{strong heat backflow}. These results show that \emph{in general} anomalous heat flow alone \emph{cannot} indicate a departure from classical explanations, in which case quantifying the anomaly is necessary to separate classical and nonclassical heat transfer. 

In this work, we show that experimental scenarios exist where any anomalous heat flow \emph{does} indicate the failure of a classical explanation. Our working definition of classicality will be the existence of a generalized noncontextual model capable of reproducing the observed data~\cite{spekkens2005contextuality}. This notion of classicality is well-motivated conceptually~\cite{spekkens2019ontological}, it emerges under quantum Darwinist dynamics~\cite{baldijao2021noncontextuality}, and subsumes other notions of classicality such as Kochen--Specker noncontextuality~\cite{budroni2022kochenspecker,kunjwal16PhD,leifer2013maximally}, ordinary classical  mechanics~\cite[Sup. Mat. A]{lostaglio2020certifying}, or non-negative quasiprobability representations~\cite{spekkens2008negativity,schmid2024structuretheorem,schmid2024kirkwooddirac}. The failure of noncontextual models to explain data can be robustly analysed~\cite{khoshbin2024alternative,selby2024linear,rossi2023contextuality,fonseca2024robustnesscontextualitydifferenttypes}, experimentally tested~\cite{mazurek2016experimental,mazurek2021experimentally,giordani2023experimental}, and   quantified~\cite{duarte2018resource,wagner2023using,catani2024resourcetheoretic}.

One of the most useful aspects of generalized noncontextuality is that, even though noncontextual models provide an intuitive classical understanding of (fragments of) physical theories, they can reproduce counterintuitive phenomena such as no-cloning~\cite{spekkens2007evidence}, teleportation~\cite{hardy1999disentangling,spekkens2007evidence}, the impossibility of discriminating nonorthogonal states~\cite{spekkens2007evidence}, and some single-photon Mach-Zehnder coherent interference patterns~\cite{catani2023whyinterference}. Even rich sub-theories of quantum theory can be framed in terms of noncontextual models such as Gaussian quantum mechanics~\cite{bartlett2012reconstruction}, or odd dimensional (single or multisystem) stabilizer subtheories~\cite{schmid2022uniqueness}. The failure of generalized noncontextuality can therefore be considered a \emph{stringent criterion} for nonclassicality. This failure is a strong and rigorous indicator that the system exhibits nonclassical behavior. Moreover, this criterion is considered stringent because contextuality is a clear, broadly applicable, and robust property of models, severely distinctive from classical behavior.

In this work, we show that anomalous heat flow implies a violation of a generalized noncontextuality inequality in an important class of experimentally meaningful scenarios. This is an indication that contextuality can be further explored as a resource for quantum thermodynamics~\cite{deffner2019quantum,binder2019thermodynamics,Goold_2016,Vinjanampathy_2016,alicki2018introduction,Kosloff_2013,Millen_2016,upadhyaya2023happens,lostaglio2018fluctuation}. In our scenario, see Fig.~\ref{fig: intro}, initial bipartite two-qubit states $\rho$, interacting via an (energy-preserving) unitary $U(t) = e^{-it H}$ during a time interval $0 < t < \tau_c$, can cause heat backflow (or increase the normal heat flow) \emph{only} when generalized noncontextual models cannot explain the statistics witnessing this property. We call $\tau_c$ the \emph{critical} time. For times $t \geq \tau_c$ outside such intervals anomaly does not necessarily imply the failure of generalized noncontextual models to reproduce the data. 

To showcase the practical relevance of our findings, we apply our results to the experiment performed in Ref.~\cite{Micadei_2019}. We also demonstrate similar results for the partial SWAP interaction between two qu\emph{dits} with dimension $d \geq 2$, showcasing that our findings do not depend on the specific Hilbert space dimension of the physical systems.

Central for demonstrating our results is the work of Ref.~\cite{lostaglio2020certifying}. We extend their findings from a specific transformation process $T$ to one where sequentially composed transformation processes $T'\circ T$ are considered. We then find that noncontextual models, restricted to obey certain operational equivalences, must be bounded by a noncontextuality inequality we construct. Our inequality recovers the one found in Ref.~\cite{lostaglio2020certifying} as a special case. We believe that our analysis of a concrete sequential scenario will have an independent interest in the program of finding noncontextuality inequalities that take into consideration the role of transformation contextuality. 

\begin{figure}[t]
    \centering
    \includegraphics[width=0.4\textwidth]{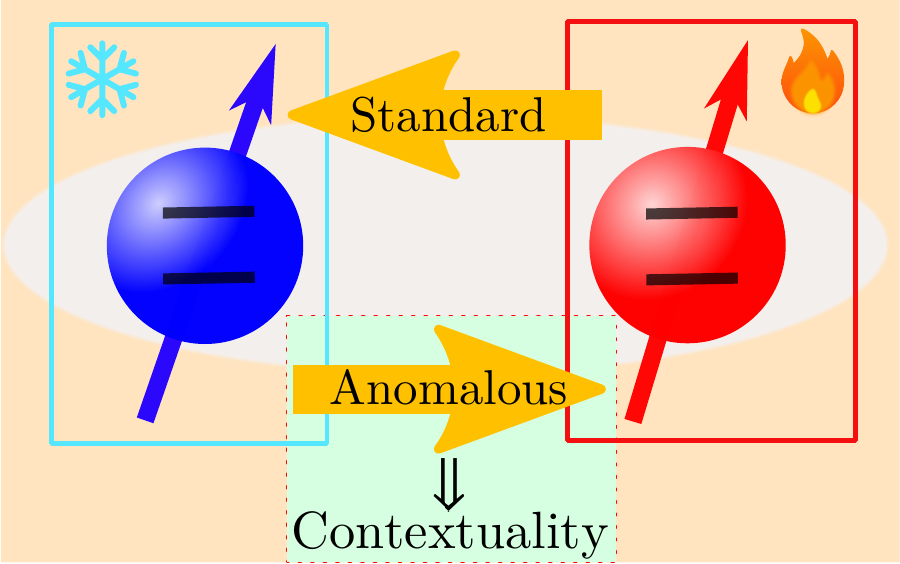}
    \caption{Sketch of our main result. A `hot' qubit is at temperature $T_A$ and a `cold' qubit is at temperature $T_B$, i.e., $T_A > T_B$. Each system individually has an associated Hamiltonian $H_A$ and $H_B$. We find that for energy preserving two-qubit interactions $U_I(t)$, i.e. satisfying $[U_I(t),H_A+H_B]=0$, any anomalous heat flow within a certain interval $0<t<\tau_c$ leads to generalized contextuality. Our results are not restricted to this simplest case of qubit systems, as we also show.}
    \label{fig: intro}
\end{figure}

This work is structured as follows: Sections~\ref{subsec_gen_setup} and~\ref{sec: clausius} describe the relevant quantum thermodynamic quantities we investigate and the possibility of reversing the usual direction of heat flow caused by initial correlations. Sec.~\ref{sec: generalized noncontextuality} reviews the concept of generalized noncontextuality and the methods developed by Ref.~\cite{lostaglio2020certifying} important to us. Section~\ref{section_3} presents our main results. We start by exposing our Theorem~\ref{theorem: sequential transformations} in which the method of Ref. \cite{lostaglio2020certifying} is extended to consider the composition of transformations. In Sec.~\ref{subsection_qubits}, we show that noncontextuality inequalities apply to a broad scenario of two interacting qubits (Theorem~\ref{theorem_3}) and show that any anomalous heat flow, for dynamical evolutions inside a critical interval $0 < t < \tau_c$, witnesses quantum contextuality. We then use our results to analyze the experiment of Ref.~\cite{Micadei_2019} and find an approximate value of $\tau_c$ for that experiment using available parameters. In Sec.~\ref{Subsection_PS_qutrits}, we consider an interaction given by a partial SWAP unitary and observe the violation of the noncontextuality inequalities for the example of heat flow between two qudits systems. In Sec.~\ref{sec: conclusion} we review our results and make our final remarks. {Unless stated otherwise, the dynamics will be evaluated in the interaction picture.}

\section{Background} \label{Sec_2}
\subsection{Average heat of local systems during an energy-conserving interaction}
\label{subsec_gen_setup}

Consider a quantum model described by the total Hamiltonian 
\begin{equation}\label{eq: main_hamiltonian}
H=H_{A}\otimes \mathbb{1}_B + \mathbb{1}_A \otimes H_{B} + H_{I},
\end{equation}
where $H_{A}$ and $H_{B}$ denote the local Hamiltonian operators for quantum systems $\mathcal{H}_A$ and $\mathcal{H}_B$, respectively, whereas $H_{I}$ is the interaction Hamiltonian operator governing the interaction between the two systems. When it is clear from our discussion, we write $H_A \equiv H_A \otimes \mathbb{1}_B$ and $H_B \equiv \mathbb{1}_A \otimes H_B$. Additionally, we suppose that the free energy is conserved during the evolution of the system~\cite{lostaglio2019introduction}. Therefore,
\begin{equation}
    [U_I(t), H_A+ H_B ] = 0, \label{eq: energy_conservation1}
\end{equation}
where $U_I(t)$ is the unitary time-evolution operator in the interaction picture, thus, given by $U_I(t) = e^{-i t H_{I}}$, setting $\hbar=1$. 

We will only consider initially correlated quantum systems $\mathcal{H}_A$ and $\mathcal{H}_B$, in a global state $\rho \in \mathcal{D}(\mathcal{H}_A \otimes \mathcal{H}_B)$, such that each local state $\rho_{i}$ ($i=A,B$) is a \emph{thermal state} (aka \emph{Gibbs states}). Here, $\mathcal{D}(\mathcal{H})$ denotes the set of all positive trace $1$ bounded operators acting on $\mathcal{H}$. This means that
\begin{align}
    \rho_i= \Tr_{\setminus \{i\}} \left\{ \rho \right\}
    = \rho_i^{\text{th}}  \equiv \frac{e^{-\beta_{i} H_{i}}}{Z_{i}}, \label{local_thermal_states}
\end{align}
where $H_{i}$ is the local Hamiltonian of $i=A,B$, $\beta_i=1/T_i$ is the inverse temperature of $i$ with Boltzmann constant $k_B=1$, and $Z_i$ is the partition function of each thermal state $Z_i = \Tr\left\{ e^{-\beta_i H_i} \right\}$.

With the above assumptions, where the Hamiltonian of the joint system is time-independent, we can define the average heat flow of the quantum system $\mathcal{H}_A$ as the total energy that it exchanges during the evolution (see ~\cite{Jennings_2010,Bera_2017,Jevtic_2012,lostaglio2020quasiprobability,Jevtic_2015,Sagawa_2012,Koski,Micadei_2019,Landi_2021} for the use of this definition)
\begin{align}
    \langle \mathcal{Q}_A \rangle & \equiv \langle {U^\dagger(t) H_A U(t)} - H_A \rangle \nonumber \\
    & = \Tr \left\{ \rho ({U^\dagger(t) H_A U(t)} - H_A) \right\}. \label{heat_A_def}
\end{align}
Notice that, with this definition, positive $ \langle \mathcal{Q}_A \rangle$ means that $\mathcal{H}_A$ \textit{receives} heat. Similarly, we can define the average heat $ \langle \mathcal{Q}_B \rangle$, of the quantum system $\mathcal{H}_B$. From energy conservation (Eq.~\eqref{eq: energy_conservation1}), we must have $ \langle \mathcal{Q}_A \rangle = -  \langle \mathcal{Q}_B \rangle$. 

\subsection{A modified Clausius inequality}\label{sec: clausius}

The dynamics described above implies that the possible average heat flow must satisfy certain inequalities, given the initial state $\rho \in \mathcal{D}(\mathcal{H}_A \otimes \mathcal{H}_B)$. If there are no initial correlations between the systems $\mathcal{H}_A$ and $\mathcal{H}_B$, so that $\rho = \rho_A^{\mathrm{th}} \otimes \rho_B^{\mathrm{th}}$ an immediate consequence is (see, for instance, Ref.~\cite[Eq. (15)]{Landi_2021}) 

\begin{equation}
   (\beta_A - \beta_B) \langle \mathcal{Q}_A \rangle  \geq 0, \label{heat_flow_uncorrelated}
\end{equation}
which is equivalent to  Clausius's statement of the second law~\cite{callen2006thermodynamics,Landi_2021,clausius1854uber,clausius1865uber} for the scenario we are considering. Eq.~\eqref{heat_flow_uncorrelated} holds because, for such a case, the left-hand side equals the entropy production, which must always be non-negative.

However, if there are initial correlations between the systems, they can be consumed for thermodynamic tasks. Consequently, the inequality above must be modified~\cite{Jennings_2010,Jevtic_2012,Jevtic_2015}. For the kind of dynamics we are interested in, just introduced in Sec.~\ref{subsec_gen_setup}, the inequality becomes 
\begin{equation}
    (\beta_A - \beta_B) \langle \mathcal{Q}_A \rangle \geq \Delta \mathcal{I}(A:B), \label{second_law_gen}
\end{equation}
where 
\begin{equation}
    \Delta \mathcal{I}(A:B) =\mathcal{I}_{U\rho U^\dagger }(A:B) - \mathcal{I}_{\rho}(A:B)
\end{equation}
is the variation of the mutual information 
\begin{equation}
    \mathcal{I}_{\rho}(A:B) = S(\rho_A) + S(\rho_B) - S(\rho),
\end{equation}
before and after the evolution, and now $\rho \in \mathcal{D}(\mathcal{H}_A \otimes \mathcal{H}_B)$ is any state for which the local states $\rho_i = \rho_i^{\mathrm{th}}$ are Gibbs states.  Albeit a well-known fact, we present a detailed derivation of the above in Appendix~\ref{Appendix_A}. 

{Inequality~\eqref{heat_flow_uncorrelated} can be derived from Ineq.~\eqref{second_law_gen}} where the parties start uncorrelated, since $\Delta \mathcal{I}(A:B)$ is always non-negative in this case. Therefore, as for the heat flow direction, both inequalities have the same meaning. For the case where the quantum system $\mathcal{H}_A$ is colder than quantum system $\mathcal{H}_B$, we \emph{must} have $\langle \mathcal{Q_A} \rangle \geq 0$, and heat flows from the hot system to the cold one, as ordinarily expected. However, for initially correlated systems, the variation of the mutual information can decrease $\Delta \mathcal{I}(A:B) < 0$, allowing for the possibility of  \textit{anomalous} heat flow $\langle \mathcal{Q_A} \rangle < 0$. This `modified Clausius inequality' includes the consumption of correlations causing heat flow inversion and still respects the second law of thermodynamics~\cite{Jennings_2010,Jevtic_2012,Jevtic_2015} (something that becomes evident when written in terms of entropy production~\cite{Landi_2021}). Note that, if we had assumed quantum system $\mathcal{H}_A$ to be hotter, then the standard heat flow would be described by $\langle \mathcal{Q}_A \rangle \leq 0$, i.e., system $\mathcal{H}_A$ loses heat, and anomaly would be characterized by $\langle \mathcal{Q}_A \rangle > 0$ indicating that system $\mathcal{H}_A$ receives heat, even {though} it is the hotter system (this will be the case considered in Subsection \ref{experimental_qubit_subsection}).

\subsection{Generalized Noncontextuality}\label{sec: generalized noncontextuality}

Generalized noncontextuality is a constraint on \emph{ontological models}~\cite{harrigan2019einstein} that attempt to explain empirical data predicted by an \emph{operational probabilistic theory} (or fragments thereof)~\cite{spekkens2005contextuality,kunjwal2019beyond,schmid2024structuretheorem}. Empirical data is obtained by acting on some system with a preparation procedure $P$, followed by a transformation $T$ and measurement $M$. Given that an outcome $k$ is obtained, the data is described statistically by some conditional probability distribution $p(k|M, T, P)$. Each procedure $P, M$, and $T$ is defined by a set of laboratory instructions to be performed. They follow a causal order (given two procedures $T_1$ and $T_2$ composed in sequence, denoted by $T_2 \circ T_1$, only the first can causally influence the second), and respect certain laws of which procedures should be applied to which systems (if a transformation $T_1$, transforms a system A to some other system $\mathrm{A}'$ and a transformation $T_2$ transforms a system B to some other system $\mathrm{B}'$ they cannot be sequentially composed as $T_2 \circ T_1$ unless $\mathrm{A}' = \mathrm{B}$; the theory needs to set rules to account for physical systems of different types)~\cite{schmid2020unscrambling,schmid2024structuretheorem}. 

Different procedures are associated with different laboratory prescriptions, but they can still yield the same data regardless of any possible operation in that theory. If this happens, these procedures lead to the same possible inferences that can be taken from the data. We then say that the two procedures are \emph{operationally equivalent}~\cite{spekkens2005contextuality} or also \emph{inferentially equivalent}~\cite{schmid2020unscrambling,rossi2023contextuality}. Formally, this defines an equivalence relation on the set of procedures in the theory: take any two transformation procedures $T_1, T_2 \in \mathcal{T}$ from the set of all possible transformations $\mathcal{T}$. These are said to be equivalent, and denoted $T_1 \simeq T_2$ if, for all conceivable preparation procedures $P \in \mathcal{P}$ and all conceivable measurements procedures $M \in \mathcal{M}$ having outcomes $k \in \mathcal{K}$,

\begin{equation}
    p(k|M,T_1,P) = p(k|M,T_2,P).
\end{equation}

We term each pair $k|M$ to be a measurement effect. Similar definitions hold for equivalent preparation procedures $P_1 \simeq P_2$ and equivalent effects $k_1|M_1 \simeq k_2|M_2$.

An ontological model~\cite{harrigan2019einstein} of an operational probabilistic theory attempts to explain the predictions of that theory in the following way: Such a model prescribes a (measurable) space $\Lambda$ of physical variables $\lambda \in \Lambda$, and assigns probability measures $\mu_P(\lambda)$ for each preparation procedure, stochastic matrices $\Gamma_T(\lambda'|\lambda)$ for each transformation procedure and response functions $\xi_{k|M}(\lambda)$ for each measurement effect $k|M$ such that they recover the empirical predictions of the operational theory from

\begin{equation}\label{eq: ontological model}
    p(k|M,T,P) = \int_\Lambda \int_{\Lambda} \xi_{k|M}(\lambda') \Gamma_T(\lambda'|\lambda)\mathrm{d}\mu_P(\lambda).
\end{equation}

An ontological model of an operational theory is said to satisfy the principle of generalized noncontextuality if it explains the operational indistinguishability of different procedures $T_1\simeq T_2$ by formally imposing that their counterparts in the model $\Gamma_{T_1}$ and $\Gamma_{T_2}$ are \emph{equal}, i.e., 
\begin{equation}
    T_1 \simeq T_2 \implies \Gamma_{T_1} = \Gamma_{T_2}.
\end{equation}
Similarly for preparation procedures and measurement effects. When there exists no generalized noncontextual model that can reproduce the data from an operational probabilistic theory, we will refer to this theory as contextual. Quantum theory, viewed as an operational theory, is contextual. This can be shown from no-go theorems~\cite{spekkens2005contextuality,banik2014ontological}, or in a quantifiable manner via the violation of generalized noncontextuality inequalities~\cite{schmid2018all,wagner2024inequalities,pusey2015robust}. 

In quantum theory, any quantum channel~\cite{nielsen2010quantum,wilde2013quantum} $\mathcal{E}: \mathcal{B}(\mathcal{H}) \to \mathcal{B}(\mathcal{H})$ defines an equivalence class of all possible physically implementable operational procedures that are described by the channel $\mathcal{E}$. $\mathcal{B}(\mathcal{H})$ denotes the bounded operators acting on $\mathcal{H}$. We write this as $\mathcal{E} = [T_\mathcal{E}]$ to represent that there can exist many different laboratory procedures $T$ described in quantum theory by the same channel $\mathcal{E}$ such that $T \simeq T_\mathcal{E}$ implying that $T \in [T_\mathcal{E}]$. Two  operationally equivalent procedures satisfy then that 
\begin{equation}
    T_1 \simeq T_2 \implies \mathcal{E}_1 = \mathcal{E}_2
\end{equation}
and the converse holds as well, but now for every element of the class represented by the quantum mechanical operators,  
\begin{equation}
    \mathcal{E}_1 = \mathcal{E}_2 \implies T_1 \simeq T_2,\, \forall T_1 \in [T_{\mathcal{E}_1}], T_2 \in [T_{\mathcal{E}_2}].
\end{equation}

Noncontextuality inequalities can be used to bound the ability of a noncontextual model to explain experimental tasks of interest, success rates of a given protocol, or other figures of merit. Commonly, such bounds are investigated in a prepare-and-measure set-up, where the role of transformations is either not considered~\cite{spekkens2009preparation,saha2019preparation,schmid2018contextual,flatt2021contextual,mukherjee2021discriminating}, or merely used to define novel preparation procedures~\cite{lostaglio2020cloningcontext,wagner2024coherence,baldijao2021noncontextuality}. In other words, most scenarios that \emph{do} consider the role of transformations do so by using transformations to define novel preparations, in the sense that if we have preparations $P \in \mathcal{P}$ and a finite set of transformations $T_1,\dots,T_n$, then we only analyze contextuality of the preparation procedures defined by new preparations $T_i(P)$, or some compositions $T_i \circ T_j (P)$, and so on. The role transformation contextuality can play in a scenario \emph{without} viewing these as defining novel preparations so far has been considered only in a handful of scenarios: to investigate nonclassicality of the stabilizer subtheory~\cite{lillystone2019contextuality,schmid2022uniqueness}, of weak values resulting from weak measurements~\cite{pusey2014anomalous,kunjwal2019anomalous} and also of quantum thermodynamics of linear response~\cite{lostaglio2020certifying}. 

\subsubsection{Noncontextuality inequalities for the average of  theory-independent observables}

In any operational theory, there is a specific transformation procedure denoted $T_{\mathrm{id}}$ that denotes the `identity' transformation. Its action is operationally characterized by
\begin{equation}
    p(k|M,T_{\mathrm{id}},P) = p(k|M,P)
\end{equation}
for all possible preparation procedures $P \in \mathcal{P}$ and all possible measurement effects $k|M \in \mathcal{K}\times \mathcal{M}$. 

In a theory-independent setup, we estimate an observable $\mathcal{A}$ by assigning some values $a_k$ to outcomes $k$ obtained once a measurement $M$ has been performed. We will then define a \emph{theory-independent observable}
\begin{equation}
    \mathcal{A} := \{(a_k,k|M)\}_k
\end{equation}
to be the \emph{finite} family of pairs of values $a_k$ and effects $k|M$, for a given measurement procedure $M$. Hence, we can \emph{define} the expectation value at an initial instant (where we haven't transformed the system of interest) given any fixed preparation $P \in \mathcal{P}$ via  
\begin{equation}
    \langle \mathcal{A}(0) \rangle := \sum_k a_k p(k|M,P)
\end{equation}
and the expectation value once a transformation took place as 
\begin{equation}
\langle \mathcal{A}(t) \rangle := \sum_k a_k p(k|M,T_t,P).
\end{equation}
Here, $t$ is merely a (suggestive) label for the transformation procedure $T_t$. Also, in all the discussion that follows, we consider, without loss of generality, that $a_k > 0$ for all $k$. The expectation for the variation between $\mathcal{A}(0)$ and $\mathcal{A}(t)$, given that preparation $P$ was performed and transformation $T_t$ proceeded it, is then given by
\begin{equation}\label{eq: average observable}
    \langle \Delta \mathcal{A}\rangle_{P,t} := \sum_k a_k \bigr(p(k|M,T_t,P) - p(k|M,P)\bigr).
\end{equation}

When it is clear which $P$ and $T_t$ are being considered we simplify the notation $\langle \Delta \mathcal{A}\rangle_{P,t} = \langle \Delta \mathcal{A}\rangle$. Ref. \cite{lostaglio2020certifying}, showed that generalized noncontextual models that attempt to reproduce the (theory-independent) averages of observables defined by Eq.~\eqref{eq: average observable} must satisfy certain inequalities that can be violated by quantum dynamics approximated by linear response perturbation theory. To be more specific, by linear response we mean that the evolution (in the interaction picture) is generated by a Hamiltonian interaction $H_I(t')$ with a weak strength parameter $g$, given in Eq.~\eqref{eq: unitary_weak_def} below, i.e satisfying $g \ll 1$. This way, the unitary evolution will have the form
\begin{equation}
    U_I(t) = \mathbb{1} - i g \int_0^t H_I(t')  dt' + O(g^2). \label{eq: unitary_weak_def}
\end{equation}
The term observable in quantum theory is usually used to denote Hermitian matrices, or more generally selfadjoint operators $A=A^\dagger$ in $\mathcal{B}(\mathcal{H})$. In this case, the theory-dependent quantum mechanical averages are given by (in the interaction picture)
\begin{equation}
    \langle \Delta A \rangle = \text{Tr}\{A(t) \rho(t)\} - \text{Tr}\{A(t)\rho\},
\end{equation}
with respect to some initial state $\rho \in \mathcal{D}(\mathcal{H})$ and where, if we are focusing on the linear response region, the interaction unitary is taken to be given by Eq.~\eqref{eq: unitary_weak_def}. Note that in quantum theory, the values $a_k$ are the \emph{eigenvalues} of the observable $A$.

An ontological model should explain the statistics of Eq.~\eqref{eq: average observable} via Eq.~\eqref{eq: ontological model}, i.e.,
\begin{align}
    \langle \Delta \mathcal{A}\rangle = \sum_k a_k \Bigr [\int_\Lambda \int_\Lambda \xi_{k|M}(\lambda')\Gamma_{T_t}(\lambda'|\lambda)\mathrm{d}\mu_P(\lambda) \nonumber \\ - \int_\Lambda \xi_{k|M}(\lambda) \mathrm{d}\mu_P(\lambda)\Bigr] \label{eq: average ontological model}
\end{align}

Ref.~\cite{lostaglio2020certifying} proved a bound for averages described in terms of  Eq.~\eqref{eq: average ontological model}, when the ontological models are generalized noncontextual, that we slightly improve in the following. 

\begin{theorem}\label{theorem: Lostaglio alpha}
Suppose that the following operational equivalence is satisfied 
\begin{equation}
     \alpha T_t + (1-\alpha)T_t^* \simeq (1-p_d)T_{\mathrm{id}} + p_d T', \label{SR_condition}
\end{equation}
where $T_t$, $T_t^*$ and $T'$ are transformation procedures, $T_{\mathrm{id}}$ is the identity procedure, and $0 \leq p_d \leq 1, 0< \alpha \leq 1$. Let $\mathcal{A}$ be an observable.  Then, any noncontextual ontological model for the variation of its average, defined by Eq.~\eqref{eq: average observable} where $a_k > 0$ for all $k$, must satisfy
\begin{equation}
    -\frac{a_{max}p_d}{\alpha^2} \leq\langle \Delta \mathcal{A} \rangle   \leq \frac{p_d a_{max}}{\alpha}. \label{main_inquality}
\end{equation}
Above, $a_{max} := \max \{a_k\}_k$ is the largest value associated with the observable  $\mathcal{A}$.
\end{theorem}

The proof follows the same reasoning given in Ref.~\cite{lostaglio2020certifying}, that we will present in detail for the proof of Theorem~\ref{theorem: sequential transformations}, and recovers the results from Ref.~\cite{lostaglio2020certifying} for the case $\alpha = 1/2$. The condition defined by Eq.~\eqref{SR_condition} is called the \emph{stochastic reversibility}, and can be understood as tossing a coin to decide if we transform the system using $T_t$ (with probability $\alpha$) or $T_t^*$ (with probability $1-\alpha$). After many realizations, the effective transformation obtained is operationally indistinguishable from doing nothing with the system with some probability $1-p_d$, and doing some other operation with probability $p_d$. The probability $p_d$ is interpreted as a `probability of disturbance' from the ideal stochastic reversibility with $p_d=0$. {For example, if we consider $T_t^*$ as inverse evolution and $\alpha= 1/2$, stochastic reversibility means that the average evolution after making an equal ensemble of forward and backward evolutions is very close to not having any evolution in the system.} We note that $T_t, T_t^*, T'$ are \emph{in principle} any triplet of transformations satisfying the operational constraint of stochastic reversibility. We assume no other constraint on such transformations. We also note that each transformation separately (at this operational level) play no fundamental role, and it is only the relation between these three operations (together with $T_{\mathrm{id}}$) that is of relevance for the characterization of the scenario. Interestingly, Ref.~\cite{lostaglio2020certifying}  suggested a method to certify if a linear response unitary satisfies
\begin{equation}
    \frac{1}{2}\mathcal{U}_t + \frac{1}{2}\mathcal{U}_t^\dagger = (1-p_d)\mathrm{id} + p_d \mathcal{C}, 
\end{equation}
where $\mathcal{U}_t(\cdot) := U(t)(\cdot)U(t)^\dagger$, with $U(t)$ given by Eq.~\eqref{eq: unitary_weak_def}, $\mathrm{id}(X) = X$ is the identity channel and $\mathcal{C}$ is some  quantum channel.

As a final remark, we would like to note that Ineq.~\eqref{main_inquality} is valid for any fragment of an operational probabilistic theory having preparation procedures $\{P_i\}_i \subseteq \mathcal{P}$, effects $\{k|M\}_k \subseteq \mathcal{M} \times \mathcal{K}$ and at least the transformations $\{T_t, T_t^*, T_t', T_{\mathrm{id}}\}_t \subseteq \mathcal{T}$, where the transformations satisfy Eq.~\eqref{SR_condition}. Here, $t$ is merely a label for the transformations in this fragment. In any such fragment, for each choice of $T_t$ and $P$, the inequality $ -p_d a_{max}/\alpha^2 \leq \langle \Delta \mathcal{A} \rangle_{P,t} \leq p_d a_{max}/\alpha$ is a valid noncontextuality inequality.

\section{Results}
\label{section_3}

We start our results by generalizing the main Theorem of Ref.~\cite{lostaglio2020certifying} (Theorem~\ref{theorem: Lostaglio alpha}, for $\alpha=1/2$) to the case where we have two sequential transformations $T_1$ and $T_2$, and under the assumption that both satisfy operational equivalences. Our main goal is to obtain a characterization of a scenario in which the presence of contextuality is due to anomalous heat flow. It is to that end that we make the aforementioned generalization, stated below in Theorem~\ref{theorem: sequential transformations}. 

\begin{theorem}\label{theorem: sequential transformations}
Let $T_t$ be a sequential transformation given in terms of other two transformation procedures $T_{t_1}$ and $T_{t_2}$, i.e., $T_t = T_{t_2} \circ T_{t_1}$. Moreover, $T_{t}$ satisfy the operational equivalences
\begin{equation}
    \frac{1}{2} T_{t} + \frac{1}{2} T_{t}^* \simeq (1- p_{d_1}) T_{t_2} + p_{d_1} T_1', \label{SR_condition_T1}
\end{equation}
where $0 \leq p_{d_1} \leq 1$, and also 
\begin{equation}
    \frac{1}{2} T_{t_2}  + \frac{1}{2} T_{t_2}^* \simeq (1- p_{d_2}) T_{\mathrm{id}} + p_{d_2} T_2', \label{SR_condition_T2}
\end{equation}
where $0\leq p_{d_2}\leq 1$. Then, any noncontextual ontological model for the average of an observable $\mathcal{A}$ must be bounded by
\begin{equation}\label{eq: sequential inequality}
    -4a_{max}b_-\leq \langle \Delta \mathcal{A} \rangle \leq 2a_{max}b_+,
\end{equation}
where $b_-:=p_{d_1}+3p_{d_2}-3p_{d_1}p_{d_2}$ and $b_+:=p_{d_1}+2p_{d_2}-2p_{d_1}p_{d_2}$, and $a_{max}:= \max \{a_k\}_k$ is the largest value associated with the observable $\mathcal{A}$.
\end{theorem}

We prove this Theorem in Appendix~\ref{Thorem_2_proof_sec}. Note that $b_-,b_+ \geq 0$ for all $0 \leq p_{d_1},p_{d_2}\leq 1$. {\color{black}It is worth noting that this result is not limited to discrete systems; as demonstrated in Appendix~\ref{Thorem_2_proof_sec}, our framework extends to scenarios involving a continuous set of ontic states, opening avenues for future application in continuous-variable systems. Consequently, this could be a theme for future research.} This Theorem will be instrumental for cases where the transformation describing the evolution of the system is not capable of satisfying Eq.~\eqref{SR_condition} but can be described as a composition of transformations satisfying this equation. Note that we recover the Theorem of Ref. \cite{lostaglio2020certifying} (or Theorem~\ref{theorem: Lostaglio alpha}, for $\alpha=1/2$) when $p_{d_2} = 0$, or also when $T_{t_2} \simeq T_{t_2}^* \simeq T_{\mathrm{id}}$. Also, when convenient we can write Ineq.~\eqref{eq: sequential inequality} more compactly as $|\langle \Delta \mathcal{A} \rangle |\leq 4a_{max}b_-$, since $b_- \geq b_+$. Relevantly, $t$ can be any label $t=(t_1,t_2)$ for which $\langle \mathcal{A}(t) \rangle = \sum_k a_k p(k|M,T_{t_2}\circ T_{t_1},P)$, and we do not assume operational equivalences to be satisfied by $T_{t_1}$. 

Similarly to Theorem~\ref{theorem: Lostaglio alpha}, our  Theorem~\ref{theorem: sequential transformations} does not depend on the fact that the operations $T_{t}, T_{t}^*$ are the inverses one from another; they work for any set of transformations satisfying these operational equivalences. Also, this result is not only valid under thermodynamic considerations: it will hold for any observable from which its theory-independent average can be described by an equation of the form Eq.~\eqref{eq: average observable}.

\subsection{Two interacting qubits} \label{subsection_qubits}

Let us consider \emph{Zeeman Hamiltonians} to be the class of single-qubit Hamiltonians defined as linear combinations of the identity operator $\mathbb{1}$ and the Pauli matrix $\sigma_z$. We now show that an interesting class of quantum interactions between two-qubit systems must satisfy the operational equivalences defined by stochastic reversibility. {The following theorem is made to obtain a practical noncontextuality inequality between \textit{any} two qubits that interact via a unitary evolution preserving their total energy, and hence can be useful and a large range of applications. Here, we shall use this for our main question, which is the influence of contextuality in the heat flow inversion.}

\begin{theorem}\label{theorem_3}
    Let $\mathcal{H}=\mathbb{C}^2\otimes \mathbb{C}^2$ be the Hilbert space describing a two-qubit system. Consider the evolution $U(t) = e^{-itH}$ with $H$ given by Eq.~\eqref{eq: main_hamiltonian}, with $H_A$ and $H_B$ Zeeman Hamiltonians. Suppose also that the interaction preserves energy,  i.e., $[H_I, H_A \otimes \mathbb{1}_B+\mathbb{1}_A \otimes H_B]=0$. Assuming an interaction picture representation of the dynamics, for every fixed instant $t$:
    \begin{enumerate}
        \item Non-resonant case ($H_A \neq H_B$): There exists a quantum channel $\mathcal{C}$ such that 
        \begin{equation}\label{eq:theorem_nonresonant}
            \frac{1}{2}\mathcal{U}_I + \frac{1}{2}\mathcal{U}_I^\dagger = (1-p_d) \mathrm{id} + p_d\mathcal{C}
        \end{equation}
        where $\mathcal{U}_I(\cdot) = U_I(t)(\cdot)U_I(t)^\dagger$, $U_I(t) = e^{-itH_I}$, and $0 \leq p_d \leq 1$.
        \item Resonant case ($H_A = H_B$): The unitary evolution $U_I(t)$ can be written as a composition of two other unitaries $U_I(t) = U_2\circ U_1$, and there exist quantum channels $\mathcal{C}_1$,  $\mathcal{C}_2$ such that
        \begin{equation}
            \frac{1}{2}\mathcal{U}_1 + \frac{1}{2}\mathcal{U}_1^\dagger = (1-p_{d_1}) \mathrm{id} + p_{d_1}\mathcal{C}_1, \label{theorem_3_resonant_eq1}
        \end{equation}
        \begin{equation}
            \frac{1}{2}\mathcal{U}_2 + \frac{1}{2}\mathcal{U}_2^\dagger = (1-p_{d_2}) \mathrm{id} + p_{d_2}\mathcal{C}_2, \label{theorem_3_resonant_eq2}
        \end{equation}
        where $\mathcal{U}_i(\cdot) = U_i(\cdot)U_i^\dagger$, and $0 \leq p_{d_i} \leq 1$ for $i=1,2$.
    \end{enumerate}
\end{theorem}

We prove this result in Appendix~\ref{Appendix_Theorem_Quantum}. We now comment on some aspects of this result that can be relevant. First, we \emph{do not} need to assume a linear response regime for the validity of the stochastic reversibility equations. The values of $p_d, p_{d_1}, p_{d_2}$ and the quantum channels $\mathcal{C}, \mathcal{C}_1, \mathcal{C}_2$ always exist. This immediately implies that such a dynamical evolution always satisfies the operational assumptions necessary for Theorem~\ref{theorem: Lostaglio alpha} and Theorem~\ref{theorem: sequential transformations}. 

We can use Theorem~\ref{theorem_3}, together with Theorem \ref{theorem: Lostaglio alpha} and Theorem~\ref{theorem: sequential transformations}, to draw the following conclusion: Let $\mathcal{H} = \mathbb{C}^2 \otimes \mathbb{C}^2$ and $U(t) = e^{-itH} = U_0(t)U_I(t)$, where $U_0(t) = e^{-it(H_A+H_B)}$ and $U_I(t) = e^{-it H_I}$ be as described in Theorem~\ref{theorem_3}. Then, \textit{any noncontextual ontological model} for a fragment of quantum theory (viewed as an operational theory) described by states $\{\rho\} \subseteq \mathcal{D}(\mathcal{H})$, POVM elements $\{E_k(t)\}_{k,t} = \{U_0(t)^\dagger E_k U_0(t)\}_{k,t} \subseteq \mathcal{B}(\mathcal{H})^+$ and unitary transformations $\{U_I(t)\}_t \subseteq \mathbb{U}(\mathcal{H})$, such that $$\Delta A_{\text{NC}} := \sum_k a_k \Big(\mathrm{Tr}\{U_I(t) \rho U_I(t)^\dagger E_k(t)\}  - \text{Tr}\{\rho E_k(t)\}\Big),$$ must satisfy the inequality 
    \begin{align}
      -4a_{max}b_- \leq \Delta A_{\text{NC}} \leq  2a_{max}b_+, \label{ineq: quantum noncontextual}
    \end{align}
where $b_-=p_{d_1}+3p_{d_2}-3p_{d_1}p_{d_2}$ and $b_+=p_{d_1}+2p_{d_2}-2p_{d_1}p_{d_2}$, for some $0 \leq p_{d_1}, p_{d_2} \leq 1$ fixed by the transformation $U(t)$ and where, $\{a_k\} \subseteq \mathbb{R}^+$ is any finite set of real positive numbers having $a_{max} := \max \{a_k\}$. {Note that Equation~\eqref{ineq: quantum noncontextual} is an inequality for when we assume the validity of quantum theory while Eq.~\eqref{eq: sequential inequality} is valid for any operational theory.} The average $\langle \Delta A(t) \rangle$ predicted by quantum theory for a positive observable $A$ is recovered when we let $E_k := \Pi^A_k $ correspond to the projections onto the eigenspace of $A$, associated with eigenvalue $a_k$. 

To show why Ineq.~\eqref{ineq: quantum noncontextual} holds, for the \textit{non-resonant} case, we use Eq.~\eqref{eq:theorem_nonresonant} as the quantum theoretical version of the operational equivalence from Theorem~\ref{theorem: Lostaglio alpha}, from which we obtain \eqref{ineq: quantum noncontextual} with $p_{d_2}=0$. As for the \textit{resonant} case, we can apply $\mathcal{U}_2$ to both sides of  Eq.~\eqref{theorem_3_resonant_eq1} and obtain a quantum theoretical version of the operational equivalence from Theorem~\ref{theorem: sequential transformations} since $\mathcal{U}_2$ is linear. This is an important use of Theorem~\ref{theorem: sequential transformations}, which can be applied in a large scope of different situations. Indeed, given \emph{any} noncontextual ontological model for a fragment of quantum theory, if two unitary transformations $\mathcal{U}_1$ and $\mathcal{U}_2$, in this fragment, satisfy Eqs.~\eqref{theorem_3_resonant_eq1} and~\eqref{theorem_3_resonant_eq2} respectively, then the procedure above can be done to obtain the Ineq.~\eqref{ineq: quantum noncontextual}.

Let us make a clarification regarding the interaction picture and the results in Theorems~\ref{theorem: Lostaglio alpha} and Theorem~\ref{theorem: sequential transformations}. Note that in the interaction picture when assuming the validity of quantum theory, it is sufficient that the unitary evolution satisfying Eq.~\eqref{SR_condition} is the one related to the \emph{interaction Hamiltonian} (regardless of energy preservation) and one simply evolves the observable with the non-interacting evolution, hence considering $E_k(t)$ instead of $E_k$. In the theory-independent setting, one can simply define the average observable to be estimated with respect to a new set of measurements $k_t|M_t$ and values $a_{k_t}$ for a fixed label $t$. This is, however, not necessary for the conclusions following from Theorems~\ref{theorem: Lostaglio alpha} and Theorem~\ref{theorem: sequential transformations}, and it is a feature of applying these findings when assuming quantum theory as an operational theory.
    
{Therefore, we have the average}
\begin{align}
    &\langle \Delta A(t) \rangle = \nonumber \\
    &= \sum_k a_k \left(\mathrm{Tr}\{U_I\rho U_I^\dagger U_0^\dagger \Pi_k^A U_0\} - \mathrm{Tr}\{\rho U_0^\dagger \Pi_k^A U_0\}\right) \nonumber
    \\&= \text{Tr}(U_I(t)\rho U_I(t)^\dagger A(t)) - \text{Tr}(\rho A(t)) \nonumber
    \\&=\text{Tr}(\rho(t)A(t))-\text{Tr}(\rho A(t)).
\end{align}
We can see why Ineq.~\eqref{ineq: quantum noncontextual} bounds noncontextual explanations for such fragments of quantum theory. In any such fragment, the dynamics characterizing any transformation procedure represented in quantum theory by the interaction unitary $U_I$ as in Theorem~\ref{theorem_3} satisfy the operational equivalences from Theorem~\ref{theorem: sequential transformations}. These equivalences are the requirements for the noncontextuality inequality from Theorem~\ref{theorem: sequential transformations} to hold for theory-independent observables given by Eq.~\eqref{eq: average observable}. When applied to any fragment of quantum theory respecting these equivalences, such observables take the form described by Ineq.~\eqref{ineq: quantum noncontextual}. As a final remark, note that in general $p_d, p_{d_1}$ and $p_{d_2}$ will depend on the parameters (and the time parameter $t$) describing the interaction Hamiltonian $H_I$. Their exact dependence is given in Appendix~\ref{Appendix_Theorem_Quantum}.

\subsubsection{Generalized contextuality without measurement incompatibility}

The above calculations allow us to draw a simple yet remarkable conclusion. We can consider a fragment of quantum theory that has (i) unitary evolutions $U(t), U(t)^\dagger$, together with the identity map and some other channel $\mathcal{C}$, for a fixed time $t$, (ii) \emph{a single measurement protocol} characterized by the effects $\{\Pi_k(t)\}_k$, and (iii) a tomographically complete set of states $\rho$ for $\mathcal{H}$. Even though such a fragment has obviously no \emph{measurement incompatibility} (as there is, in fact, a single measurement) such a fragment may lead to a violation of the inequality above (as we will show later). Indeed, that proofs of the failure of generalized noncontextuality can be found with a single measurement is a theoretical prediction from Ref.~\cite{selby2023incompatibility}. Our construction provides a concrete fragment of quantum mechanics in which this can be verified by the violation of a noncontextuality inequality, obtained from Theorem~\ref{theorem: sequential transformations}. 

\subsubsection{{Contextuality as a necessary condition for anomalous heat flow}}

We will now show that contextuality allows for anomalous quantum heat transfer beyond what \emph{any} noncontextual model is capable of. Using Theorem \ref{theorem_3}, we now study violations of the generalized noncontextuality obtained when there is heat exchanged between two qubits, $\mathcal{H}_A = \mathbb{C}^2 = \mathcal{H}_B$. They are described locally in terms of Zeeman Hamiltonians which, without loss of generality, we chose to be $H_A = \frac{\omega_A}{2}(\mathbb{1}-\sigma_z) $ and $ H_B = \frac{\omega_B}{2}(\mathbb{1}-\sigma_z)$, and they interact via a energy-preserving unitary. From the definition of Eq.~\eqref{heat_A_def}, we see that if we choose $A:= H_A \otimes \mathbb{1}_B$, then 
\begin{equation}
    \langle \mathcal{Q}_A \rangle =  \langle A(t) - A(0)\rangle = \langle H_A \otimes \mathbb{1}_B(t)-H_A \otimes \mathbb{1}_B \rangle,
\end{equation}
and we can investigate the average heat flow from quantum system $\mathcal{H}_B$ to system $\mathcal{H}_A$ (or vice-versa). From Theorem~\ref{theorem_3}, the noncontextuality inequalities (Eqs.~\eqref{main_inquality} and~\eqref{eq: sequential inequality}) bound the heat received by $A$ explainable by noncontextual models. Inequality~\eqref{main_inquality} bounds the non-resonant case while inequality~\eqref{eq: sequential inequality} bounds the resonant case.

For the non-resonant case $\omega_A \neq \omega_B$ the study is trivial since no heat is transferred. This happens because by a direct calculation, for all times 
\begin{equation}
    \langle \mathcal{Q}_A \rangle = 0.
\end{equation}
It is simple to see that this holds because the most general non-resonant interaction Hamiltonian $H_{I_{\mathbf{nr}}}$ in this case that is capable of satisfying $[H_{I_{\mathbf{nr}}},H_A + H_B] = 0$ must have the form  $H_{I_{\mathbf{nr}}} = g (\vert 01\rangle \langle 01 \vert + \vert 10\rangle \langle 10 \vert )$ (see Appendix \ref{Appendix_Theorem_Quantum}).

For the resonant case $\omega_A=\omega_B=\omega$, and
\begin{equation}
    H_A = \frac{\omega}{2} (\mathbb{1} - \sigma_z^A). \label{local_qubit_Hamiltonian}
\end{equation}
We now consider the most general case of an initial two-qubit density matrix (in the eigenbasis of $\sigma_z^A \otimes \sigma_z^B$)  written as
\begin{equation}
   \rho = \begin{pmatrix}
        p_{00} & \nu_1^* & \nu_2^* & \gamma^* \\
        \nu_1  &  p_{01} & \eta e^{i \xi} & \nu_3^* \\
        \nu_2 & \eta e^{-i \xi} &  p_{10} & \nu_4^* \\
        \gamma & \nu_3  & \nu_4  &  p_{11}
    \end{pmatrix}, 
\end{equation}
where $0 \leq p_{00}$, $p_{01}$, $p_{10}$, $p_{11} \leq 1$, $\eta$ and $\xi$ are real numbers, and the remaining parameters are complex numbers, constrained to satisfy that $\rho$ is positive semidefinite and  $p_{00}+p_{01}+p_{10}+p_{11} = 1$.  Note that we choose a commonly used basis representation to make our calculations, but our results are basis-independent, as is usually true for proofs of the failure of noncontextuality~\cite{wagner2024inequalities}.  Here, $\nu^*$ denotes the complex conjugate of the complex number $\nu$. The most general energy-conserving unitary interaction for the resonant case is given by $U_{I_\text{r}} (t)= e^{-i H_{I_\text{r}} t}$, where $H_{I_\text{r}}$ is a Hamiltonian given by 
\begin{equation}
    H_{I_\text{r}}= g \begin{pmatrix}
        0 & 0 & 0 & 0 \\
        0 & a & e^{i \theta} & 0 \\
        0 & e^{-i \theta} & a & 0 \\
        0 & 0 & 0 & 0
    \end{pmatrix}, \label{qubit_general_energy_conserving_Hamiltonian}
\end{equation}
where $g>0$ and $0 \leq \theta \leq 2\pi$. As demonstrated in the Appendix \ref{Appendix_Theorem_Quantum} in the proof of Theorem~\ref{theorem_3}, we find that $U_{I_{\mathrm{r}}} = U_2\circ U_1$ and the unitary evolutions characterized by such Hamiltonians satisfy the operational equivalences shown in Theorem~\ref{theorem_3} for $p_{d_1} = \sin^2[(a-1)gt/2]$ and $p_{d_2} = \sin^2(gt)$. Note that in both cases $p_{d_i} = O(g^2t^2)$ when $gt \ll 1$.

For this situation, a direct computation of the average heat (Eq.~\eqref{heat_A_def}) results in
\begin{equation}
{    \!\!\langle \mathcal{Q}_A \rangle \! = \! \omega \! \left[ (p_{01}\!-\!p_{10})\sin^2(gt) \!+\! \eta \sin(2gt) \sin(\xi\!-\!\theta) \right].}\label{general_heat_2qubits}
\end{equation}
This equation asserts that, for the case of resonant qubits,  transformation contextuality happens for sufficiently small $gt$  \textit{if and only if} the presence of coherence in the initial density matrix has some effect on the heat flow. This is because, using Theorem \ref{theorem_3} and Theorem \ref{theorem: sequential transformations} with $A:= H_A \otimes \mathbb{1}_B$, the absolute value of the heat must be bounded by a quantity of order $O(g^2t^2)$, to allow for a noncontextual explanation of the dynamics (Ineq.~\eqref{eq: sequential inequality}). However, the heat contribution caused by coherence, namely, the term $\eta \sin(2gt) \sin(\theta +\xi)$, is $O(gt)$, which can always violate the noncontextual bound for small enough $gt$. Ref.~\cite{lostaglio2020certifying} called this the contextuality of quantum linear response.

For small enough $gt$, whenever the variation of the observable (in this case, the energy) is $O(gt)$, a quantum violation of the noncontextuality inequality will always exist. More concretely, we can find the range for values $gt$ in which anomaly must imply quantum contextuality.  Since in an experimental setup, $g$ is normally a fixed parameter, this means that we are interested in finding the interval $0<t<\tau_c$, for some \emph{critical time} $\tau_c$ such that any anomalous heat flow cannot be explained by noncontextual models. Note that in general, $t=0$ starts with no heat flow, and we let $t=\tau_c$ be the regime where we lose a violation of the noncontextuality inequality~\eqref{eq: sequential inequality}. 

While our main focus has been on the earliest critical time $\tau_c$ at which contextuality due to anomalous heat flow ceases to be guaranteed, this is by no means the \textit{only} interval in which noncontextuality inequalities can be violated. In the most general setting, our results imply that, for certain evolutions, multiple disjoint time windows exist in which a violation of the noncontextuality inequality signals anomalous heat flow. For concreteness---and to facilitate comparison with experimental data---we then concentrate on this specific first time interval.

If we impose for the initial state $\rho$ that the local states are thermal with inverse temperatures $\beta_A$ and $\beta_B$ (see Eq.~\eqref{local_thermal_states}) then the most general global state $\rho$ will have the form
\begin{align}
    \rho = \begin{pmatrix}
        \nu_0 & \nu_1^* & \nu_2^* & \gamma^* \\
        \nu_1  &  \frac{1}{Z_A} - \nu_0 & \eta e^{i \xi} & -\nu_2^* \\
        \nu_2 & \eta e^{-i \xi} &  \frac{1}{Z_B} - \nu_0 & -\nu_1^* \\
        \gamma & -\nu_2  & -\nu_1  &  \frac{e^{-\omega \beta_A - \omega \beta_B}-1}{Z_A Z_B} + \nu_0
    \end{pmatrix}, \label{general_rho_2qubits_thermal}
\end{align}
where $\nu_0$ is a real number and $Z_i = 1+e^{-\omega \beta_i}$, $i=A,B$. This is the case considered in the Sec.~\ref{subsec_gen_setup}, with two qubit systems. The average heat  (Eq.~\eqref{general_heat_2qubits}), after some manipulations, becomes 
\begin{align}
    \langle \mathcal{Q}_A \rangle & = \omega \Bigg\{ \frac{1}{2}\sin^2(gt)\left[ \tanh\left(\frac{\omega \beta_A}{2}\right) - \tanh\left(\frac{\omega \beta_B}{2}\right) \right] \nonumber \\
    & + \eta \sin(2gt) \sin(\xi-\theta) \Bigg\}. \label{general_heat_2qubits_thermal}
\end{align}

The term 
{$(\omega/2)\sin^2(gt)\left[ \tanh\left(\omega \beta_A/ 2\right)\!-\!\tanh\left(\omega \beta_B/2\right) \right]$}
in the equation above is of order $O(g^2t^2)$, and is responsible for the standard heat flow in the absence of correlations. For instance, if $T_A < T_B$, then $\beta_A = 1/T_A > 1/T_B = \beta_B$ and therefore  
{$\tanh\left(\omega \beta_A/2\right) > \tanh\left(\omega \beta_B/2\right)$}, meaning that quantum system $\mathcal{H}_A$, in this case the colder system, receives heat as expected.

Concurrently, the term $\omega \eta \sin(2gt) \sin(\xi-\theta)$ has sign and magnitude depending only on the coherence term $\eta e^{i\xi}$ and the interaction parameters $g,~\theta$. This allows the latter term to be capable of causing anomalous heat flow or to increase the standard heat flow depending on the type of coherence and interactions. Importantly, this last term is of order $O(gt)$, which (from Theorem 3) implies that, whenever this term is non-zero, there will be a contextuality witness for small enough $gt$, or equivalently for \emph{some} interval of time $0 < t < \tau_c$. Since any heat flow inversion is possible only due to this last term, we conclude that:

\begin{corollary} \label{corollary_1}
For any two qubits (having local Zeeman Hamiltonians, and associated Gibbs states) interacting via an energy-preserving unitary, anomalous heat flow is possible for $0<t<\tau_c$ only if noncontextual models fail to explain the data.
\end{corollary}

This corollary is our main result. It shows a contextuality signature for a fairly broad class of thermodynamic scenarios. It is also important to note that the term $\nu_0$ in Eq.~\eqref{general_rho_2qubits_thermal} carries all the information about the incoherent correlations in the initial density matrix of the two qubits. However, this term has no role in the heat exchanged (Eq.~\eqref{general_heat_2qubits_thermal}). This leads us to conclude that \emph{only} coherent correlations and entanglement are responsible for the anomalous heat flow, for the interacting two-qubits case. This relation between heat flow anomaly and the aforementioned quantum resources was already known from Ref.~\cite{lostaglio2020quasiprobability}. We add the statement that it is both these resources, together with the dependency on $gt$ and the validity of the operational equivalences we investigate, that make these correlations inexplicable by means of noncontextual ontological models.

{
In addition, we can point out a relation between the violation of the standard Clausius inequality (Eq.~\eqref{heat_flow_uncorrelated}) and the presence of contextuality when $gt \ll 1$. That is, for short times, the
same heat flow term contributing to the anomaly, and hence to the violation to the noncon-
textuality inequality, is responsible for the initial negativity of the mutual information change
in Eq.~\eqref{second_law_gen} (see Appendix \ref{qubits_mutual_information}). This suggests a direct relationship between contextuality and the emergence of correlation effects in the second law of thermodynamics.
}

\subsubsection{Connection with experimental results}
\label{experimental_qubit_subsection}

We claimed that for fixed values of $g$ there are critical times $\tau_c$ that can be found, such that for any $0<t<\tau_c$ the noncontextual inequality is violated for this class of two-qubit interactions we have just studied. Small $gt$ bounds the amount of time in which this contextuality witness happens in real experiments, and we wish this time interval to be experimentally accessible. Therefore, we now estimate such a critical time $\tau_c$ using the parameters from the recent experimental result of Ref.~\cite{Micadei_2019}.

\begin{figure}[t]
    \centering
   \includegraphics[width = \columnwidth]{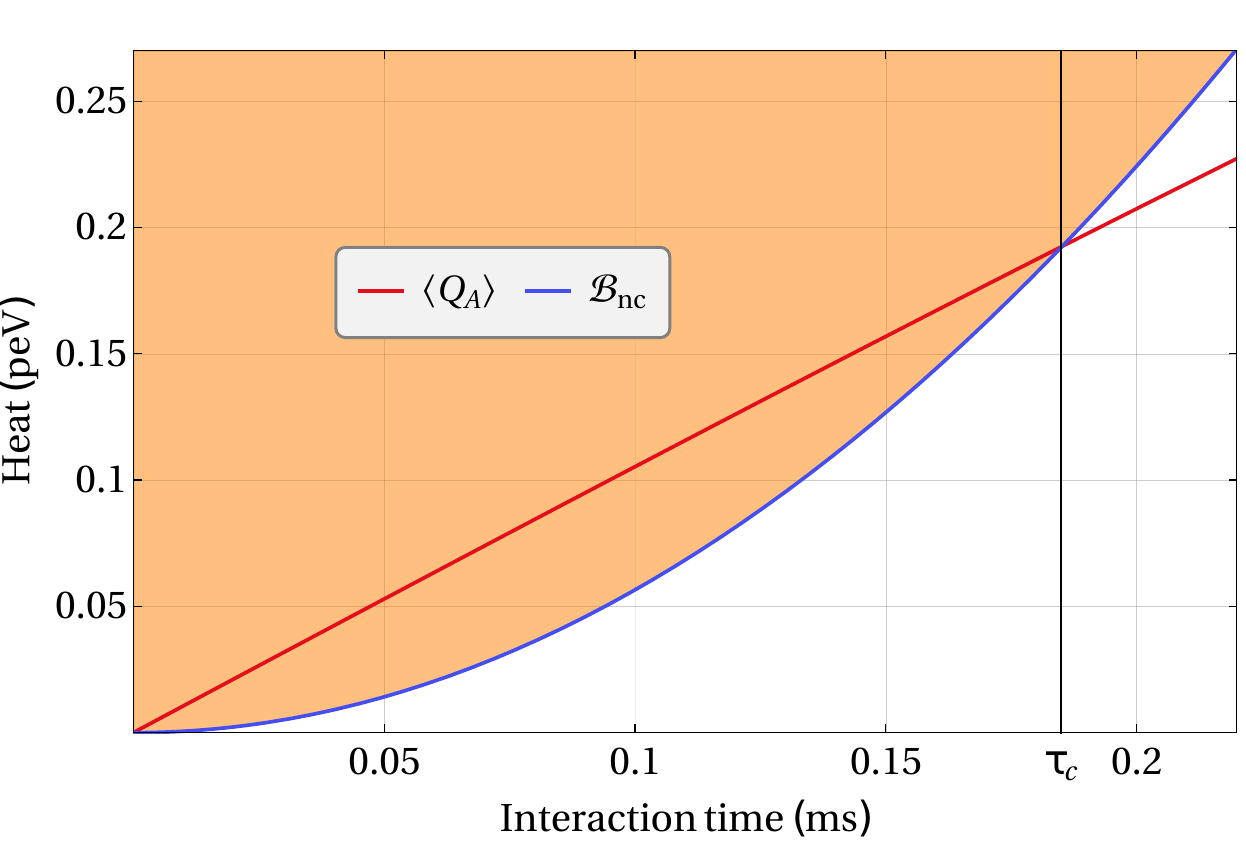}
\caption{\textbf{Anomalous heat flow with $0<t<\tau_c$ implies a violation of a noncontextuality inequality.} Average heat flow predicted by quantum theory (red line) and noncontextuality bound (blue curve). The noncontextuality bound is given by  Ineq.~\eqref{eq: sequential inequality}. The region above the curve (orange) corresponds to the values in which the heat transfer \emph{cannot} be described by a noncontextual model. Any anomalous heat flow violates the noncontextuality inequality for $0 < t < \tau_c \approx 1.85\times 10^{-4}$s. In this interval, heat averages predicted by quantum theory are greater than those achievable by noncontextual models. Interaction Hamiltonian given by Eq.~\eqref{V_interaction_2qubits_exp}, and initial quantum state given by Eq.~\eqref{general_rho_2qubits_thermal}. Parameters: $J = 215.1$ Hz,  $\omega = 4.135\times 10^{-12}~ \text{eV}$, $T_A = 4.3$ peV, $T_B = 3.66$ peV, $\gamma = \xi = \nu_0 = \nu_1 = \nu_2 = 0$, $g= J \pi$, $a=0$, $\theta=\pi/2$ and $\eta = -0.19$, taken from Ref.~\cite{Micadei_2019}. (color online)} 
\label{plot_heat_exp}
\end{figure}

In this experiment, the two qubits are realized by spin-1/2 systems in a nuclear magnetic resonance (NMR) setup and are resonant with local Hamiltonians in the form of Eq.~\eqref{local_qubit_Hamiltonian} with $\omega=h\nu_{\text{exp}}$, where $h \approx 4.135\times 10^{-15}~ \text{eV}.\text{Hz}^{-1}$ is the Plank constant and $\nu_\text{exp} = 1$ kHz. The interaction is set by the following effective interaction Hamiltonian
\begin{equation}
    H_I = \frac{\pi J}{2} (\sigma_x^A \otimes \sigma_y^B - \sigma_y^A \otimes \sigma_x^B), \label{V_interaction_2qubits_exp}
\end{equation}
where $J = 215.1$ Hz. The initial temperature (in units of energy) of the qubit $\mathcal{H}_A$ is $T_A = 4.3$ peV, and of the qubit $\mathcal{H}_B$ is $T_B = 3.66$ peV (recalling, $T_i = \beta_i^{-1}$). Note therefore that in their experiment $\mathcal{H}_A$ is the \emph{hotter} system, so any anomalous heat flow happens when this system \emph{receives heat}, which in our convention is characterized by $\langle \mathcal{Q}_A \rangle > 0$.

Additionally, the experiment considers the initial state in the form of Eq.~\eqref{general_rho_2qubits_thermal} with $\gamma = \xi = \nu_0 = \nu_1 = \nu_2 = 0$, and $\eta = -0.19 $. This situation corresponds to the resonant case with heat given by Eq.~\eqref{general_heat_2qubits_thermal} and an interaction Hamiltonian having parameters $g = J \pi,~a = 0$, and $\theta = \pi/2$.

Theorem~\ref{theorem: sequential transformations} implies that the absolute value of the heat received by the qubit $\mathcal{H}_A$ must be bounded by the function 
\begin{align}
\mathcal{B}_\text{nc} 
:&= 
2a_{max}b_+ = 2a_{max}(p_{d_1}+2p_{d_2}-2p_{d_1}p_{d_2}) \nonumber\\
&=
2h \nu_0[\sin^2(J\pi t) + 2 \sin^2(J\pi t/2) 
\nonumber
\\
&\hspace{2cm}
- 2\sin^2(J\pi t)\sin^2(J \pi t/2)]
\nonumber\\
&{~= 2h \nu_0[1 - \cos^{3}{(J\pi t)}]},
\end{align}
defining the noncontextual bound. In Fig.~\ref{plot_heat_exp} we see the region where the heat flow predicted by quantum theory (red curve), given by Eq.~\eqref{general_heat_2qubits_thermal} with the parameters of the experiment of Ref. \cite{Micadei_2019}, violates such bound for small $t$ (the anomalous heat flow here is positive since $T_A > T_B$, a standard heat flow would be negative). Therefore, we obtain the approximate value for critical time to be $\tau_c \approx 1.85 \times 10^{-4}$ s. 

{
As mentioned previously, during the reversal of heat flow, the initial correlations are `consumed', and as a result, the variation of mutual information $\Delta \mathcal{I}(A :B)$ is negative. A clear image for this consumption is presented in Figure 2 of Micadei et al.~\cite{Micadei_2019} for this experimental case. As we show in Appendix~\ref{qubits_mutual_information}, this consumption has a negative slope boosted by the same heat flow term that violates the noncontextuality inequality.
}

We can also see that, from Fig.~\ref{plot_heat_exp}, when $t\geq\tau_c$ anomalous heat flow can happen, yet without the violation of the noncontextuality inequality (it is important to stress that respecting this noncontextuality inequality is a necessary but not a sufficient condition for noncontextuality). We can also comment that the plot in Fig.~\ref{plot_heat_exp} is zoomed in so that the values of $\langle \mathcal{Q}_A \rangle$ appear to be a line when they are described by a periodic function. This implies that in fact there are various critical time intervals $\Delta\tau_{ij} = (\tau_{c}^i,\tau_{c}^j)$ in which the reversal of heat flow witnesses quantum contextuality. Finally, while we have focused on the linear regime due to the experimental parameters, \emph{in principle} different choices of parameters could allow for the anomaly to witness contextuality \emph{beyond} a linear response regime.

\subsection{{Two interacting qudit systems: the partial SWAP as a case study}} \label{Subsection_PS_qutrits}

Our main theoretical result from Theorem~\ref{theorem: sequential transformations} is not dependent on the dimensionality of the physical system considered. To highlight this fact, we explore similar conclusions as in the two qubits case, using the methods of Theorem~\ref{theorem: Lostaglio alpha} or Theorem~\ref{theorem: sequential transformations}, {by finding interactions that respect the stochastic reversibility condition (Eq.~\eqref{SR_condition}) valid to any bipartite system $\mathbb{C}^d \otimes \mathbb{C}^d$ for some dimension $d$}. { In what follows we make statements valid for any Hilbert space $\mathcal{H}$, and specify to the case of dynamics governed by the partial SWAP acting on $\mathcal{H} = \mathbb{C}^d \otimes \mathbb{C}^d$. For concreteness, we conclude with an analysis of the case $d=3$. }


\subsubsection{A simple  family of unitary evolutions satisfying stochastic reversibility}

We say that a Hermitian operator $K=K^\dagger$ is also \emph{involutory} when $K^2 = \mathbb{1}$. From Stones theorem, any family of unitary operators $\{U(t)\}_t$ can be written as $U=e^{-itK}$ for some $K$ self-adjoint, and for the cases when $K$ is also involutory we can show that it satisfies stochastic reversibility.

\begin{proposition}\label{proposition: involutory}
    Let $\mathcal{H}$ be any  Hilbert space. Let $K$ be an involutory selfadjoint operator from $\mathcal{B}(\mathcal{H})$ and $U_K(t) := e^{-igtK}$. Then, $\mathcal{U}_K (\cdot) := U_K(t)(\cdot)U_K(t)^\dagger $ satisfies the stochastic reversibility condition~\eqref{SR_condition}, with $\mathcal{C}(\cdot) = K (\cdot) K$, $\alpha=1/2$ and $p_d = \sin^2(gt)$.
\end{proposition}

\begin{proof}
    Note that  $\mathcal{C}(X) = K X K^\dagger$ and $K^\dagger K = KK = \mathbb{1}$ and hence $\mathcal{C}$ is in a Kraus representation form, and therefore it is a quantum channel. The validity of stochastic reversibility with this specific $p_d$ follows from Taylor expansion since $K$ is bounded and involutory.
\end{proof}

Examples of Hermitian involutory matrices are the SWAP operator, the Pauli matrices, the identity matrix, CNOT gates, conjugations~\cite{miyazaki2022imaginarityfree}, and the Hadamard matrix. 

\subsubsection{Partial SWAP}

A unitary that satisfies this condition and energy conservation (Eq.~\eqref{eq: energy_conservation1}) is the \emph{partial SWAP} defined as 
\begin{equation}
    U_{PS}(t) := e^{-i g t S},
\end{equation}
where $S$ is the SWAP operator~\cite{nielsen2000quantum}, and $g > 0$. Since the SWAP operator is described by an involutory matrix, i.e., $S^2 = \mathbb{1}$, Proposition~\ref{proposition: involutory} applies to $e^{-i g t S}$, where 
\begin{equation}
    U_{PS}(t) = \cos(gt) \mathbb{1} - i \sin(gt) S. \label{u to general swap}
\end{equation}
By direct calculation, we have 
\begin{equation}
    \frac{1}{2}\mathcal{U}_{PS} + \frac{1}{2}\mathcal{U}_{PS}^\dagger = \cos^2(gt) \mathrm{id} + \sin^2(gt) \mathcal{C}_S,
\end{equation}
where $\mathcal{U}_{PS}[X] := U_{PS}(t) X U_{PS}^\dagger(t)$ and $\mathcal{U}^\dagger_{PS}[X] = U_{PS} (t)^\dagger X^\dagger U_{PS} (t)$, $p_d = \sin^2(gt)$ and the quantum channel $\mathcal{C}_S[X] = S X S$. { Note that the partial SWAP dynamics is valid for any $S$ acting on $\mathbb{C}^d \otimes \mathbb{C}^d$. In what follows, we investigate heat flow in this general case.} 

{
\subsubsection{Heat flow during a partial SWAP of two qudits}
The partial SWAP unitary can be useful to obtain noncontextuality inequalities for higher dimensional systems since it satisfies the stochastic reversibility condition in any dimension where $S$ is well-defined. An example that may find fruitful applications is any two-qudits interacting via a partial SWAP. We find conclusions similar to the two-interacting qubits of Subsection \ref{subsection_qubits} for these cases.

Consider two single-qudit quantum systems $\mathcal{H}_A$ and $\mathcal{H}_B$ each with a Hilbert space of dimension $d \geq 2$. The qudit $\mathcal{H}_A$ has a local Hamiltonian 
\begin{equation}
    H_A = \sum_{k=0}^{d-1} \omega_k \ket{k}_A \bra{k}_A, \label{qudit_A_local_hamiltonian}
\end{equation}
where $\omega_k$ are positive numbers (the so-called {Bohr frequencies}) and $\ket{k}_A$ the eigenvectors of $H_A$. If the global state of the composite Hilbert space  $\mathcal{H}_A \otimes \mathcal{H}_B$ is initially 
\begin{equation}
    \rho = \sum_{n,m,n',m'} p_{n,m,n',m'} \ket{n,m}\bra{n',m'}, \label{initial_state_qudit_AB}
\end{equation}
where $p_{n,m,n',m'}$ are the density matrix elements respecting $\rho >0$ and $\tr \{ \rho \} = 1$, with marginals $p_{n,n'}^A = \sum_{m} p_{n,m,n',m}$ and $p_{m,m'}^B = \sum_{n} p_{n,m,n,m'}$. Note that we have defined $\ket{n,m} := \ket{n}_A \otimes \ket{m}_B$, where $\ket{n}_A$ is an eigenvector of $H_A$ with eigenvalue $\omega_n$ and $\ket{m}_B$, is an eigenvector of the local Hamiltonian acting on the Hilbert space of $B$. Given these conditions, we show (see Appendix \ref{Appendix_qudits}) that the heat \emph{received} by $\mathcal{H}_A$, after interacting with $\mathcal{H}_B$ via a partial SWAP (given by Eq. \eqref{u to general swap}) is 
\begin{align}
   \langle \mathcal{Q}_A \rangle & = \sin^2 (g t) \left( \sum_m p_{m,m}^B \omega_m - \sum_n p_{n,n}^A \omega_n \right) \nonumber \\
   & - \sin (2 g t) \sum_{n,m} \text{Im}(p_{n,m,m,n}) \omega_n. \label{heat_qudit_PSWAP}
\end{align}

This equation maintains the same characteristics as the heat exchanged between any two qubits (Eq. \eqref{general_heat_2qubits}), and, consistently, it is equivalent to Eq.~\eqref{general_heat_2qubits} for $d=2$ and $\theta=0$. The following three aspects are of main importance: First, the first term of the rhs of the equation above, namely $\sin^2 (g t) \left( \sum_m p_{m,m}^B \omega_m - \sum_n p_{n,n}^A \omega_n \right)$, which depends only on the diagonal terms of the global density matrix (and, more specifically, only on the diagonal terms of the reduced density matrices), is responsible for the standard heat flow. For instance, if the local states are thermal $\rho_A = \frac{e^{-\beta_A H_A}}{Z_A}$ and $\rho_B = \frac{e^{-\beta_B H_B}}{Z_B}$, for $\beta_{A(B)}$ inverse temperatures, $Z_{A(B)}$ partition functions, and the local Hamiltonians $H_{A(B)} = \sum_{k=0}^{d-1} \omega_k \ket{k}_{A(B)} \bra{k}_{A(B)}$, with $\omega_k = k \Delta \omega > 0$ (equally spaced Böhr frequencies), this first term will be

\begin{align}
   & \Delta \omega \frac{\sin^2 (g t)}{2} \Bigg[ \coth(\beta_B \Delta \omega/2) - d \coth (d \beta_B \Delta \omega/2)  \nonumber \\
    & - \left( \coth(\beta_A \Delta \omega/2) - d \coth (d \beta_A \Delta \omega/2) \right) \Bigg],
\end{align}
which is always positive for $\beta_A > \beta_B$, meaning that it contributes positively to the heat received by $\mathcal{H}_A$ when the quantum system  $\mathcal{H}_A$ is locally in a colder state than the quantum system $\mathcal{H}_B$. Clearly, this term is also consistent with the case of two qubits (for $d=2$ and $\theta =0$, this gives the first term of the rhs of Eq.~\eqref{general_heat_2qubits_thermal}). 

Second, the second term of the rhs of Eq.~\eqref{heat_qudit_PSWAP}, namely $- \sin (2 g t) \sum_{n,m} \text{Im}(p_{n,m,m,n}) \omega_n$, is the term responsible for the possibility of anomalous heat flow and depends only on the coherent terms of the initial density matrix (in the energy eigenbasis). 

Third, this second term of the rhs of Eq.~\eqref{heat_qudit_PSWAP} is of type ${O}(gt)$ for small values of $gt$, while the first term (responsible for the standard heat flow) is of type ${O}(gt)$. Therefore, this second coherent term is always capable of violating a non-contextuality inequality for a sufficiently small $gt$. This follows directly from Proposition~\ref{proposition: involutory} (which is valid for the partial SWAP unitary) and Theorem~\ref{theorem: Lostaglio alpha}. 

Hence, we showed that the density matrix coherence terms can cause an effect on the heat flow exchanged between two qudits, interacting via a partial SWAP, \emph{if and only if} a non-contextuality inequality can be violated. Also, for this case, the anomalous heat flow happens \emph{only if} a noncontextuality inequality can be violated. These facts happen in complete analogy with the heat exchanged between two qubits of Subsec.~\ref{subsection_qubits}, and show that the same conclusions can be extended for a case of higher dimension. 

Let us investigate the violation of the noncontextuality inequality (Theorem~\ref{theorem: Lostaglio alpha}) using Eq.~\eqref{heat_qudit_PSWAP}. Let $\omega_{max} = \max \{\omega_k\}_{k=0,d-1}$. Then, we can find the critical times $\tau_c^l,\tau_c^u$ assuming that $\omega_{max}, g$ are all fixed parameters and $\sum_{n,m} \text{Im}(p_{n,m,m,n}) \omega_n \neq 0$, given by
\begin{align}
    \tau_c^u &= \frac{1}{g}\cot^{-1}\left(\frac{2\omega_{max}-\left( \sum_m p_{m,m}^B \omega_m - \sum_n p_{n,n}^A \omega_n \right)}{ 2\sum_{n,m} \text{Im}(p_{n,m,m,n}) \omega_n}\right),\label{eq: tau upper u}\\
    \tau_c^l &= \frac{1}{g}\cot^{-1}\left(\frac{-4\omega_{max}-\left( \sum_m p_{m,m}^B \omega_m - \sum_n p_{n,n}^A \omega_n \right)}{ 2\sum_{n,m} \text{Im}(p_{n,m,m,n}) \omega_n}\right), \label{eq: tau lower l}
\end{align}
where above we have considered $\tau_c^u$ the critical time associated with the noncontextuality inequality upper bound, and $\tau_c^l$ for the lower bound. These bounds can be used to directly infer the relationship between contextuality and anomalous heat flow. For a concrete example, we will focus on $d=3$ in what follows.

}

\subsubsection{{Exemplification: Heat flow during a partial SWAP of two qutrits}}

The partial SWAP unitary can be useful to obtain noncontextuality inequalities for higher dimensional systems since it satisfies the stochastic reversibility condition in any dimension where $S$ is well-defined. An intriguing model for examination involves the interplay between two qutrits mediated by a SWAP gate. This framework finds utility in elucidating experimental implementations of quantum gates and quantum engines, while concurrently facilitating extensions to qudit systems~\cite{scovil1959three, garcia2013swap, wilmott2011swapping}. Each qutrit manifests solely three perfectly distinguishable states, denoted as $\ket{0}$, $\ket{1}$, and $\ket{2}$, that we write as the canonical basis
\begin{equation}
    \ket{0} = \begin{pmatrix}
             1 \\ 0 \\ 0  
            \end{pmatrix}, \quad 
            \ket{1} = \begin{pmatrix}
             0 \\ 1 \\ 0 \\ \end{pmatrix}\quad \text{and} \quad 
             \ket{2} = \begin{pmatrix}
             0 \\ 0 \\ 1  
        
            \end{pmatrix}.
\end{equation}
It is ensured that the composite state of the system adheres to the form $\ket{i}\otimes\ket{j}$ ($i,j=0,1,2$). Consequently, the SWAP operator $S_3$ is characterized by the conditions
\begin{equation}
    S_3\ket{i}\otimes\ket{j} = \ket{j}\otimes\ket{i},\quad S_3^2 = \mathbb{1}, 
\end{equation}
Thus, the SWAP operator assumes the form
\begin{align}
    S_3 = \begin{pmatrix}
             1 & 0 & 0 & 0 & 0 & 0 & 0 & 0 & 0\\
             0 & 0 & 0 & 1 & 0 & 0 & 0 & 0 & 0\\
             0 & 0 & 0 & 0 & 0 & 0 & 1 & 0 & 0\\
             0 & 1 & 0 & 0 & 0 & 0 & 0 & 0 & 0\\
             0 & 0 & 0 & 0 & 1 & 0 & 0 & 0 & 0\\
             0 & 0 & 0 & 0 & 0 & 0 & 0 & 1 & 0\\
             0 & 0 & 1 & 0 & 0 & 0 & 0 & 0 & 0\\
             0 & 0 & 0 & 0 & 0 & 1 & 0 & 0 & 0\\
             0 & 0 & 0 & 0 & 0 & 0 & 0 & 0 & 1         
            \end{pmatrix}
\end{align}
and the unitary operator governing the system's dynamics arising from this interaction is expressed as (see Eq.~\eqref{u to general swap})
\begin{equation}
    U_3 (t) = e^{-igtS_3} = \cos(gt)\mathbb{1} - i\sin (gt)S_3. \label{U for qutrit}
\end{equation}
As before, $g$ denotes the strength of the interaction. The local Hamiltonian governing a single qutrit system is expressed as
\begin{equation}
    H_A = \omega_0 \ket{0}\bra{0}+ \omega_1 \ket{1}\bra{1} + \omega_2 \ket{2}\bra{2}, \label{qutrit hamiltonian}
\end{equation}
where $\omega_i \in \mathbb{R}^+$ denotes the energy associated with eigenstate $\ket{i}$,{ this Hamiltonian is the realization of the qudit Hamiltonian of Eq.~\eqref{qudit_A_local_hamiltonian}}. Also, we consider both local Hamiltonians to be equal, i.e., $H_A = H_B$. Consequently, a density matrix, with thermal local states (Eq.~\eqref{local_thermal_states}), describing the composite system while conserving energy is generically represented as
\begin{widetext}
\begin{align}\label{eq: state qutrit}
    \rho_{3 \otimes 3} = \begin{pmatrix}
             p_0 & 0 & 0 & 0 & 0 & 0 & 0 & 0 & 0\\
             0 & p_1 & 0 & \eta_{31}e^{i\theta_{31}} & 0 & 0 & 0 & 0 & 0\\
             0 & 0 & p_2 & 0 & 0 & 0 & \eta_{62}e^{i\theta_{62}} & 0 & 0\\
             0 & \eta_{31}e^{-i\theta_{31}} & 0 & p_3 & 0 & 0 & 0 & 0 & 0\\
             0 & 0 & 0 & 0 & p_4 & 0 & 0 & 0 & 0\\
             0 & 0 & 0 & 0 & 0 & p_5 & 0 & \eta_{75}e^{i\theta_{75}} & 0\\
             0 & 0 & \eta_{62}e^{-i\theta_{62}} & 0 & 0 & 0 & p_6 & 0 & 0\\
             0 & 0 & 0 & 0 & 0 & \eta_{75}e^{-i\theta_{75}} & 0 & p_7 & 0\\
             0 & 0 & 0 & 0 & 0 & 0 & 0 & 0 & p_8        
            \end{pmatrix},
\end{align}
\end{widetext}
where the off-diagonal elements denote correlations between the qutrits {(note that the nonzero terms are the ones relevant to the heat received by $A$, according to Eq. \eqref{heat_qudit_PSWAP})}and the diagonal elements, denoted as $\{p_i\}_i$, form a classical probability distribution contingent upon the inverse temperatures $\beta_A$ and $\beta_B$ of qutrits $\mathcal{H}_A$ and $\mathcal{H}_B$, respectively (refer to Appendix~\ref{Appendix_C} for their specific formulations). Note that the parameters $\eta_{31},\eta_{62},\eta_{75},\theta_{31},\theta_{62},\theta_{75}$ are \emph{not} free, and must be such that $\rho_{3 \otimes 3}$ is positive semidefinite, together with the fact that $\{p_i\}_{i=0}^8$ must define a probability distribution. The heat content of qutrit $\mathcal{H}_A$, expressed as $\langle \mathcal{Q}_A \rangle$, can be evaluated { directly} from Eq.~\eqref{heat_A_def}   { or using Eq. \eqref{heat_qudit_PSWAP}}  
\begin{align}
     \langle \mathcal{Q}_A \rangle &= \Tr \left\{ \rho_{3 \otimes 3} (U_3 (t) H_A \otimes \mathbb{1}_B U^\dagger_3 (t) - H_A\otimes \mathbb{1}_B) \right\}\nonumber \\
    &= \zeta \sin^2(gt) + \xi\sin (gt)\cos(gt) , 
    \nonumber \\
    &{= \zeta \sin^2(gt) + \frac{\xi}{2}\sin (2gt) ,}
    \label{heat for qutrit}
\end{align}
where $\zeta$ { is equivalent to the sum $\sum_m p_{m,m}^B \omega_m - \sum_n p_{n,n}^A \omega_n$ of Eq. \eqref{heat_qudit_PSWAP} and} depends on the inverse temperatures $\beta_i$ and energies $\omega_i$ of the qutrits, while $\xi$ { is equivalent to the sum $\sum_{n,m} \text{Im}(p_{n,m,m,n}) \omega_n/2$ of Eq. \eqref{heat_qudit_PSWAP} and} encapsulates the correlations $\eta_{ij}$ and phases $\theta_{ij}$ (see Appendix \ref{Appendix_C} for details). { Exemplifying our discussion on qudits,} the term $\zeta$ is responsible for the standard heat flow (meaning that $\beta_A < \beta_B$ implies $\zeta>0$), while $\xi$ is the term that makes the heat flow inversion possible. For small values of $gt$, contextual effects are assured for non-null values of $\xi$. It is noteworthy that while $\xi$ can be utilized to reverse the heat flux, such inversion is not a prerequisite for contextuality, since this term can also cause an increase in the standard heat flow. Let us investigate the violation of the noncontextuality inequality (found in Ref.~\cite{lostaglio2020certifying}) using Eq.~\eqref{heat for qutrit}. Let $\omega_{max} = \max \{\omega_0,\omega_1,\omega_2\}$. Then, we can find the critical times $\tau_c^l,\tau_c^u$ assuming that $\zeta, \xi, \omega_{max}, g$ are all fixed parameters and $\xi \neq 0$, given by
\begin{align}
    \tau_c^u &= \frac{1}{g}\cot^{-1}\left(\frac{2\omega_{max}-\zeta}{\xi}\right),\label{eq: tau upper}\\
    \tau_c^l &= \frac{1}{g}\cot^{-1}\left(\frac{-4\omega_{max}-\zeta}{\xi}\right), \label{eq: tau lower}
\end{align}
where above we have considered $\tau_c^u$ the critical time associated with the noncontextuality inequality upper bound, and $\tau_c^l$ for the lower bound. These bounds can be used to directly infer the relationship between contextuality and anomalous heat flow.

{Moreover, the change of the mutual information in this situation has an analogous result as in the two-qubits case. Again, the heat term responsible for the anomalous heat flow and the violation of our noncontextuality inequality is the same that causes the change of the mutual information to be negative for $gt \ll 1$ (see Appendix~\ref{Appendix_C} for details). Hence, the effects of correlations on the Clausius inequality also have a direct connection with contextuality certification in this case, which can indicate the same pattern to the higher-dimensional case.
}

\section{Summary and outlook}\label{sec: conclusion}
In this paper, we demonstrated that heat flow inversion, caused by initial correlations between two-qubit unitary interactions that conserve total energy \textit{cannot} be described by generalized noncontextual models when the interaction happens in certain intervals of time $0<t<\tau_c$. These results introduce the notion of critical times $\tau_c$ governing dynamical nonclassicality. 

For the kind of qubit interactions we have investigated, a coherent effect in the heat flow (either reversing the flow or increasing it in the standard direction) must be present to allow for a violation of the noncontextual inequalities we have introduced. We also show analogous results for {two qudits, with any dimension $d$,} interacting via a partial SWAP, { indicating that similar conclusions in higher dimensional systems are true for a large class of applicable interactions, and can be further explored for more general interactions}.  

As an application of our findings, we use the experimental parameters of Ref.~\cite{Micadei_2019} to show that the connection we find between contextuality and reversal of the spontaneous direction of heat flow can be tested by existing quantum hardware and state-of-the-art manipulation of quantum resources. For an experimental test, a robust account on the critical time $\tau_c$ could be found by extending Theorems~\ref{theorem: Lostaglio alpha},~\ref{theorem: sequential transformations}, and~\ref{theorem_3} to consider experimental imperfections. 

It is worth pointing out that Ref.~\cite{lostaglio2020quasiprobability} showed connections between the reversal of the direction of heat flow and another form of nonclassicality, defined to be the negativity of the real part of Kirkwood-Dirac quasiprobability distributions~\cite{wagner2024circuits,lostaglio2022kirkwood,arvidssonshukur2024properties,halpern2018quasiprobability,gherardini2024quasiprobabilities}, known as the Terletsky-Margenau-Hill quasiprobability distribution~\cite{margenau1961correlation,terletsky1937limiting,halpern2018quasiprobability,lostaglio2020quasiprobability}. As they pointed out, negativity in such distributions is a proxy for generalized contextuality~\cite{spekkens2008negativity,schmid2024structuretheorem,schmid2022uniqueness}, but known to \emph{not} be sufficient for contextuality in general~\cite{schmid2024kirkwooddirac}. 

Furthermore, one can also use the combination of Theorem \ref{theorem: sequential transformations} and Theorem \ref{theorem_3} to explore the violation of noncontextuality inequalities for the variation of other observables, rather than energy, for energy-conserving interactions between qubits. For instance, one can study the { population inversion} by choosing the observable $A = \sigma_z$ (for density matrices in the eigenbasis of $\sigma_z$) and obtain noncontextuality inequalities in this case. This scenario is useful for certifying the presence of contextuality in open quantum systems with the collisional models approach~\cite{Rau_PhysRev.129.1880,Campbell_2021_collisional,Ciccarello_2022}, where the models often involve two-qubits interactions (see \cite{Buzek_PhysRevLett.88.097905,Ziman_PhysRevA.65.042105,Ciccarello_2022,Cusumano,Comar_PhysRevA.104.032217,Campbell_2019} for examples).

Similarly to Ref.~\cite{Jennings_2010} which showed that the occurrence of \textit{strong heat backflow} is a signature of entanglement, our results argue that the presence of heat flow inversion caused by correlations -- or, more generally, any coherent correlation effect in the heat flow -- in energy-conserving two-qubits interactions{, and partial SWAP interactions between two qudits,} imply a signature of generalized contextuality for a time interval between the beginning of the interaction and a critical time $\tau_c$, with the value of $\tau_c$ depending on the parameters of the scenario. Similar scenarios are ubiquitous in quantum thermodynamics~\cite{deffner2019quantum,binder2019thermodynamics,Goold_2016,Vinjanampathy_2016,alicki2018introduction,Kosloff_2013,Millen_2016}. We believe our results can be useful to certify contextuality in a variety of models as well as to indicate contextuality as a resource for genuinely nonclassical phenomena in quantum thermodynamics.

\section*{Acknowledgments}
We acknowledge Matteo Lostaglio for his valuable suggestions, including the estimation of the critical time $\tau_c$ in comparison with experimental parameters. We also acknowledge the suggestions from PRX Quantum reviewers, which, among other things, led us to the results concerning to the two interacting qudits via a Partial SWAP. Furthermore, we acknowledge Rodrigo Isaú Ramos Moreira for helping us with the SWAP operator for qudits. 

BA and NEC acknowledge financial support from Instituto Serrapilheira, Chamada n.~4 2020. {\color{black}NEC also acknowledges the support of the Financiadora de Estudos
e Projetos (grant 1699/24 IIF-FINEP)}. BA and DC acknowledge Pró-Reitoria de Pesquisa e Inovação (PRPI) from the Universidade de São Paulo (USP) for financial support through the Programa de Estímulo à Supervisão de Pós-Doutorandos por Jovens Pesquisadores. BA also acknowledges Fundação de Amparo à Pesquisa do Estado de São
Paulo, Auxílio à Pesquisa - Jovem Pesquisador, grant number 2020/06454-7. 
RW acknowledges support from FCT -- Fundação para a Ciência e a Tecnologia (Portugal) through PhD Grant SFRH/BD/151199/2021 and from the European Research Council (ERC) under the European Union’s Horizon 2020 research and innovation programme (grant agreement No. 856432, HyperQ). LFS acknowledges the financial support of Coordenação de Aperfeiçoamento de Pessoal de Nível
Superior (CAPES) – Brazil, Finance Code 001, and the financial support of São Paulo Research Foundation (FAPESP) under grant number 2024/08114-0. This work was also supported by the Digital Horizon Europe project FoQaCiA, GA no.101070558, funded by the European Union, NSERC (Canada), and UKRI (U.K.). 

\bibliography{references}

\appendix
\onecolumngrid

\section{Proof of Equation~\eqref{second_law_gen}}
\label{Appendix_A}

This proof follows the pedagogical compendium of Ref.~\cite{Landi_2021}. Assume the conditions presented in Sec.~\ref{subsec_gen_setup}. Let $\rho' \equiv U \rho U^\dagger$ and $\rho'_i := \Tr_{\setminus \{i\}}\{\rho'\}$. The quantity 
\begin{equation}
    \mathcal{S} =  S(\rho_A'|| \rho_A) + S(\rho_B'|| \rho_B), \label{S_definition}
\end{equation}
is always non-negative since it is a sum of relative entropies $S(\rho||\sigma) = \Tr \left\{ \rho \log(\rho) - \rho \log(\sigma) \right\} \geq 0$. Now, we can rewrite 
\begin{align}
    S(\rho_A' || \rho_A ) & =  \Tr\left\{ \rho_A' \log(\rho_A') - \rho_A' \log(\rho_A) \right\} \nonumber \\
    & = S(\rho_A) - S(\rho_A') + \Tr\left\{\rho_A \log(\rho_A) - \rho_A' \log(\rho_A) \right\} \nonumber \\
    & = -\Delta S_A + \Tr\left\{\rho_A \log(\rho_A) - \rho_A' \log(\rho_A) \right\}, \nonumber
\end{align}
where we used the definition of the von Neumann entropy $S(\rho) = -\Tr \{ \rho \log(\rho) \}$ and defined $\Delta S_A = S(\rho_A') - S(\rho_A)$.

Furthermore, using Eq.~\eqref{local_thermal_states}, we have $\log (\rho_A) = -\beta_A H_A - \log(Z_A)$. Using this in the equation above, we obtain
\begin{align}
    S(\rho_A' || \rho_A ) & =  -\Delta S_A + \beta_A \Tr \left\{ (\rho_A'-\rho_A)H_A \right\} \nonumber \\
    & = -\Delta S_A + \beta_A \langle Q_A \rangle, \label{relative_ent_A}
\end{align}
where in the last equation we used Eq.~\eqref{heat_A_def}. Analogously, we have 
\begin{equation}
    S(\rho_B' || \rho_B ) = -\Delta S_B + \beta_B \langle Q_B \rangle, \label{relative_ent_B}
\end{equation}
where $\Delta S_B = S(\rho_B') - S(\rho_B)$.

Given the mutual information $\mathcal{I}_\rho(A:B) = -S(\rho) + S(\rho_A) + S(\rho_B)$ between the systems $\mathcal{H}_A$ and $\mathcal{H}_B$, for a unitary evolution, the variation of the mutual information $\Delta\mathcal{I}(A:B) = \mathcal{I}_{\rho'}(A:B) - \mathcal{I}_{\rho}(A:B)$ will be
\begin{equation*}
    \Delta\mathcal{I}(A:B) = \Delta S_A +  \Delta S_B,
\end{equation*}
since $S(\rho) = S(\rho')$. Therefore, by summing Eqs.~\eqref{relative_ent_A} and~\eqref{relative_ent_B} and using the equation above, we obtain
\begin{align}
    \mathcal{S} = &  S(\rho_A'|| \rho_A) + S(\rho_B'|| \rho_B) \nonumber \\
    = & \beta_A \langle \mathcal{Q}_A \rangle + \beta_B \langle \mathcal{Q}_B \rangle - \Delta \mathcal{I}(A:B) \geq 0, \label{second_law_gen_deduction}
\end{align}
which implies 
\begin{equation}
    (\beta_A - \beta_B) \langle \mathcal{Q}_A \rangle \geq \Delta \mathcal{I}(A:B),
\end{equation}
where we used energy conservation $\langle \mathcal{Q}_A \rangle = - \langle \mathcal{Q}_B \rangle$. 

Notice that, in order to deduce Eq.~\eqref{second_law_gen}, we have only used that the unitary evolution does not change the von Neumann entropy, energy conservation, and $\mathcal{S} \geq 0$. This points out that, to enunciate a second law, we may only need to find a suitable always non-decreasing `irreversible' quantity, which, in this case, was $\mathcal{S}$. The fact that this quantity is always non-negative is a mathematical fact. However, it implies some bounds on the physical quantities given the specific evolution studied.

\section{Proof of Theorem~\ref{theorem: sequential transformations}}\label{Thorem_2_proof_sec}

In what follows we simply write $T_{t_1} \equiv T_1$ and $T_{t_2} \equiv T_2$ for simplicity.  We start by considering the operational equivalence
\begin{align*}
    \frac{1}{2} T_2 \circ T_{1} + \frac{1}{2} T_2 \circ T_{1}^* \simeq (1- p_{d_1}) T_2 + p_{d_1} T_2 \circ T_1'
\end{align*}
that we re-write as
\begin{equation}
  \frac{1}{2}T + \frac{1}{2}T_{A_1}  \simeq (1-p_{d_1}) T_2 + p_{d_1} T_{A_2}. \label{op_equiv_lemma2}
\end{equation}
where $T=T_2\circ T_1$, $T_{A_1}  = T_2 \circ T_1^*$ and $T_{A_2} = T_2 \circ T'$. Our proof will follow the exact same methodology as the one of Theorem 1 in Ref.~\cite{lostaglio2020certifying}. We devide the proof into two cases, where we first consider ontological models where $\Lambda$ is a finite set of ontic states, and later we consider a more involved proof for the case where $\Lambda$ is a continuous measurable set of ontic states.

\textit{Finite set $\Lambda$ of ontic states.}-- Consider the ontological models description for the average  difference $\langle \Delta \mathcal{A}\rangle$ in the observable $\mathcal{A} = \{(a_k,k|M)\}_k$ given by
\begin{equation}
  \langle \Delta \mathcal{A}\rangle = \sum_k a_k \bigr(p(k|M,T,P)-p(k|M,P)\bigr) = \sum_k a_k \left( \sum_{\lambda,\lambda'} \mu_P (\lambda) \Gamma_{T} (\lambda'|\lambda) \xi_{M}(k|\lambda') - \sum_\lambda \mu_P(\lambda) \xi_{M}(k | \lambda) \right), \label{op_variation_ont_model_finite}
\end{equation}
where $a_k>0$ are the possible values assigned to $\mathcal{A}$ from effects  $k|M$, $\mu_P(\lambda)$ is the probability distribution associated by the model to the preparation procedure $P$, $\Gamma_{T} (\lambda'|\lambda)$ is the transition matrix representing the probability of the ontic states to change from $\lambda$ to $\lambda'$, and $\xi_{M}(k | \lambda)$ is the response function associated to $k|M$ which gives the probability distribution of obtaining the outcome $k$, given that the ontic state was $\lambda$ and that a measurement $M$ was performed. Without loss of generality, we consider $\Lambda = \Lambda'$.

Since $\Gamma_{T} (\lambda'|\lambda) \leq 1, \forall \lambda,\lambda' \in \Lambda$ in the equation above, we have
\begin{equation}
    \langle \Delta \mathcal{A}\rangle \leq \sum_k a_k \left( \sum_{\lambda\neq \lambda'} \mu_P (\lambda) \Gamma_{T} (\lambda'|\lambda) \xi_{M}(k|\lambda') \right).  \label{lemma2_ineq_step1}
\end{equation}

We now consider that the ontological models must respect Eq.~\eqref{op_equiv_lemma2}, constraining the model to be  noncontextual, implying that for all $\lambda,\lambda' \in \Lambda$ 
\begin{equation}
\frac{1}{2} \Gamma_{T} (\lambda'|\lambda) + \frac{1}{2}\Gamma_{A_1} (\lambda'|\lambda) = (1-p_{d_1}) \Gamma_{T_2} (\lambda'|\lambda)+ p_{d_1} \Gamma_{A_2} (\lambda'|\lambda). \label{noncotext_condition_lemma2}
\end{equation}

Moreover, the assumption of transformation noncontextuality applied to Eq.~\eqref{SR_condition_T2} (stochastic reversibility for $T_{2}$), also requires that for all $\lambda,\lambda' \in \Lambda$
\begin{equation}
    \frac{1}{2} \Gamma_{T_2} (\lambda'|\lambda) + \frac{1}{2} \Gamma_{T_2^*} (\lambda'|\lambda) = (1- p_{d_2}) \delta_{\lambda',\lambda} + p_{d_2} \Gamma_{C_2}(\lambda'|\lambda), \label{noncotext_condition_lemma2_T2}
\end{equation}
where $\delta_{\lambda, \lambda'} = \Gamma_{T_{\mathrm{id}}}(\lambda|\lambda')$ is the transition matrix representation of the identity transformation by the ontological model. Using the two equations above, and the fact that the transition matrices are always non-negative, we obtain 
\begin{equation*}
 \Gamma_{T} (\lambda'|\lambda) \leq 4(1-p_{d_1})(1-p_{d_2}) \delta_{\lambda, \lambda'} + 4(1-p_{d_1})p_{d_2} \Gamma_{C_2} (\lambda'|\lambda) + 2 p_{d_1} \Gamma_{A_2} (\lambda'|\lambda). 
\end{equation*}
Using this inequality in the inequality~\eqref{lemma2_ineq_step1}, we obtain 
\begin{align}
    \langle \Delta \mathcal{A}\rangle  &\leq \sum_k a_k \left( \sum_{\lambda\neq \lambda'} \mu_P (\lambda) \Bigr(4(1-p_{d_1})p_{d_2} \Gamma_{C_2} (\lambda'|\lambda) + 2 p_{d_1} \Gamma_{A_2} (\lambda'|\lambda)\Bigr) \xi_{M}(k|\lambda') \right) \nonumber \\
    &\leq \sum_k a_k \left( \sum_{\lambda, \lambda'} \mu_P (\lambda) \Bigr(4(1-p_{d_1})p_{d_2} \Gamma_{C_2} (\lambda'|\lambda) + 2 p_{d_1} \Gamma_{A_2} (\lambda'|\lambda)\Bigr) \xi_{M}(k|\lambda') \right). \label{lemma2_ineq_step2}
\end{align}
Where in the last inequality we have summed back the terms $\lambda = \lambda'$, that are always non-negative.

Now we notice that the term $\sum_{\lambda, \lambda'} \mu_P (\lambda) \Gamma_{C_2} (\lambda'|\lambda) \xi_{M}(k|\lambda') = p(k|M, T_{C_2}, M)$ is just a probability distribution of having an outcome $k$ under a given evolution, therefore
\begin{equation}
    \sum_k a_k \sum_{\lambda, \lambda'} \mu_P (\lambda) \Gamma_{C_2} (\lambda'|\lambda) \xi_{M}(k|\lambda') = \sum_k a_k p(k|M, T_{C_2}, M) \leq a_{max} \sum_k p(k|M, T_{C_2}, M) = a_{max}, \label{a_max_argument_lemma2}
\end{equation}
since $a_{max}\geq a_k, \forall k$ and $\sum_kp(k|M,T,P) =1, \forall T \in \mathcal{T}, P \in \mathcal{P}$. The same  is also valid for the term $\sum_{\lambda, \lambda'} \mu_P (\lambda) \Gamma_{A_2} (\lambda'|\lambda) \xi_{M}(k|\lambda')= p(k|M, T_{A_2}, M)$, and therefore using these inequalities in the inequality~\eqref{lemma2_ineq_step2}, we obtain
\begin{align}
    \langle \Delta \mathcal{A}\rangle \leq  \Bigr(4(1-p_{d_1})p_{d_2} + 2 p_{d_1} \Bigr) a_{max} = 2(p_{d_1} + 2p_{d_2} - 2p_{d_1}p_{d_2}) a_{max}. \label{lemma2_first_ineq}
\end{align}

To complete the proof for the finite $\Lambda$ case, we analyze the converse. From equation Eq.~\eqref{op_variation_ont_model_finite} we have
\begin{equation*}
    -\langle \Delta \mathcal{A}\rangle = - \langle \mathcal{A}(t) - \mathcal{A}(0) \rangle = - \sum_k a_k \left( \sum_{\lambda,\lambda'} \mu_P (\lambda) \Gamma_{T} (\lambda'|\lambda) \xi_{M}(k|\lambda') - \sum_\lambda \mu_P(\lambda) \xi_{M}(k | \lambda) \right).
\end{equation*}
Proceeding similarly as before, isolating $\Gamma_T$ using Eqs.~\eqref{noncotext_condition_lemma2} and~\eqref{noncotext_condition_lemma2_T2} and substituting it in the definition of $\langle \Delta \mathcal{A}\rangle$ we obtain
\begin{align}
     &- \langle \Delta \mathcal{A}\rangle = - \sum_k a_k \Bigg( \sum_{\lambda,\lambda'} \mu_P (\lambda) \Big[ -2(1-p_{d_1})\Gamma_{T_2^*} (\lambda'|\lambda) - \Gamma_{A_1} (\lambda'|\lambda) + 4(1-p_{d_1})(1-p_{d_2}) \delta_{\lambda, \lambda'} \nonumber \\
   & +4(1-p_{d_1})p_{d_2} \Gamma_{C_2} (\lambda'|\lambda) + 2p_{d_1} \Gamma_{A_2} (\lambda'|\lambda) \Big] \xi_{M}(k|\lambda') - \sum_\lambda \mu_P(\lambda) \xi_{M}(k|\lambda) \Bigg) \nonumber \\
   &\leq \sum_k a_k \Bigg( \sum_{\lambda,\lambda'} \mu_P (\lambda) \left[ 2(1-p_{d_1})\Gamma_{T_2^*} (\lambda'|\lambda) + \Gamma_{A_1} (\lambda'|\lambda) - 4(1-p_{d_1})(1-p_{d_2}) \delta_{\lambda, \lambda'} \right] \xi_{M}(k|\lambda')   + \sum_\lambda \mu_P(\lambda) \xi_{M}(k | \lambda) \Bigg) \nonumber \\
   = & ~  \sum_k a_k \Bigg[  2(p_{d_1}+2p_{d_2}-2p_{d_1}p_{d_2})  \sum_{\lambda} \mu_P (\lambda) \xi_{M}(k|\lambda) + 2(1-p_{d_1})\left( \sum_{\lambda,\lambda'}  \mu_P (\lambda)\Gamma_{T_2^*} (\lambda'|\lambda)\xi_{M}(k|\lambda') - \sum_{\lambda} \mu_P (\lambda) \xi_{M}(k|\lambda) \right) \nonumber \\
   & + \left( \sum_{\lambda,\lambda'}  \mu_P (\lambda)\Gamma_{A_1} (\lambda'|\lambda)\xi_{M}(k|\lambda')  - \sum_{\lambda} \mu_P (\lambda) \xi_{M}(k|\lambda) \right) \Bigg], \label{big_computation_finite_case_Lemma2}
\end{align}
where in the inequality we used the fact that the terms $ -\sum_k a_k \sum_{\lambda,\lambda'} \mu_P (\lambda)4(1-p_{d_1})p_{d_2} \Gamma_{C_2} (\lambda'|\lambda)\xi_{M}(k|\lambda') $ and $-\sum_k a_k \sum_{\lambda,\lambda'} \mu_P (\lambda) 4p_{d_1} \Gamma_{A_2} (\lambda'|\lambda)\xi_{M}(k|\lambda')$ are never positive, and in the last equality we summed over the Kronecker delta and redistributed the terms. Now, notice that the term $\sum_k a_k \left( \sum_{\lambda,\lambda'}  \mu_P (\lambda)\Gamma_{A_1} (\lambda'|\lambda)\xi_{M}(k|\lambda') - \sum_{\lambda} \mu_P (\lambda) \xi_{M}(k|\lambda) \right)$ is analogous to the r.h.s. of Eq.~\eqref{op_variation_ont_model_finite} by exchanging $\Gamma_{T} (\lambda'|\lambda)$ for $\Gamma_{A_1} (\lambda'|\lambda)$, and from Eqs.~\eqref{noncotext_condition_lemma2} and~\eqref{noncotext_condition_lemma2_T2}, we have that
\begin{equation*}
    \Gamma_{A_1} (\lambda'|\lambda) \leq 4(1-p_{d_1})(1-p_{d_2}) \delta_{\lambda, \lambda'} + 4(1-p_{d_1})p_{d_2} \Gamma_{C_2} (\lambda'|\lambda) + 2 p_{d_1} \Gamma_{A_2} (\lambda'|\lambda),
\end{equation*}
from which, analogous to the deduction of the inequality~\eqref{lemma2_first_ineq}, we obtain
\begin{equation}
    \sum_k a_k \left( \sum_{\lambda,\lambda'}  \mu_P (\lambda) \Gamma_{A_1}(\lambda'|\lambda)\xi_{M}(k|\lambda') - \sum_{\lambda} \mu_P (\lambda) \xi_{M}(k|\lambda) \right) \leq 2(p_{d_1} + 2p_{d_2} - 2p_{d_1}p_{d_2}) a_{max}. \label{A_1_last_step_ineq}
\end{equation}
In a similar way, the term $\sum_k a_k \left( \sum_{\lambda,\lambda'}  \mu_P (\lambda)\Gamma_{T_2} (\lambda'|\lambda)\xi_{M}(k|\lambda') - \sum_{\lambda} \mu_P (\lambda) \xi_{M}(k|\lambda) \right)$ is also analogous to the r.h.s. of Eq.~\eqref{op_variation_ont_model_finite}, and from Eq.~\eqref{noncotext_condition_lemma2_T2}, we obtain
\begin{equation*}
    \Gamma_{T_2} (\lambda'|\lambda) \leq 2(1- p_{d_2}) \delta_{\lambda',\lambda} + 2p_{d_2} \Gamma_{C_2}(\lambda'|\lambda),
\end{equation*}
from which, using the same arguments that leads to~\eqref{lemma2_first_ineq}, we obtain
\begin{equation}
    \sum_k a_k \left( \sum_{\lambda,\lambda'}  \mu_P (\lambda) \Gamma_{T_2}(\lambda'|\lambda)\xi_{M}(k|\lambda') - \sum_{\lambda} \mu_P (\lambda) \xi_{M}(k|\lambda) \right) \leq 2 p_{d_2} a_{max}. \label{T_2_last_step_ineq}
\end{equation}
Using Ineqs.~\eqref{A_1_last_step_ineq} and~\eqref{T_2_last_step_ineq} in Ineq.~\eqref{big_computation_finite_case_Lemma2}, we have
\begin{align}
    - \langle \Delta \mathcal{A}\rangle \leq  2(p_{d_1}+2p_{d_2}-2p_{d_1}p_{d_2}) \sum_k a_k \sum_{\lambda} \mu_P (\lambda) \xi_{M}(k|\lambda) \nonumber \\
    + 4(1-p_{d_1})p_{d_2} a_{max} + 2(p_{d_1} + 2p_{d_2} - 2p_{d_1}p_{d_2}) a_{max}. \label{proof_lemma2_ineq_last_step_finite}
\end{align}
Finally, noticing that the term $\sum_{\lambda} \mu_P (\lambda) \xi_{M}(k|\lambda) = p(k|M,P)$, we have as before $\sum_k a_k \sum_{\lambda} \mu_P (\lambda) \xi_{M}(k|\lambda) = \sum_k a_k p(k|M,P) \leq a_{max}.$ Using this in Ineq.~\eqref{proof_lemma2_ineq_last_step_finite}, we obtain
\begin{align}
    - \langle \Delta \mathcal{A}\rangle & \leq  2(p_{d_1}+2p_{d_2}-2p_{d_1}p_{d_2}) a_{max} + 4(1-p_{d_1})p_{d_2} a_{max} + 2(p_{d_1} + 2p_{d_2} - 2p_{d_1}p_{d_2}) a_{max} \nonumber \\
    & = 4(p_{d_1}+3 p_{d_2} - 3 p_{d_1}p_{d_2}) a_{max},
\end{align}
which completes the proof for the finite $\Lambda$ case. 

\textit{Continuous set $\Lambda$ of ontic states.}-- For an \emph{infinite} (and possibly \emph{continuous)} set $\Lambda$, we have the following version of Eq.~\eqref{op_variation_ont_model_finite}
\begin{equation}
  \langle \Delta \mathcal{A}\rangle = \langle \mathcal{A}(t) - \mathcal{A}(0) \rangle = \sum_k a_k \left( \int_\Lambda d \lambda' \int_\Lambda d \lambda ~ \mu_P (\lambda) \Gamma_{T} (\lambda'|\lambda) \xi_{M}(k|\lambda') - \int_\Lambda d\lambda ~ \mu_P(\lambda) \xi_{M}(k | \lambda) \right), \label{op_variation_ont_model_infinite}
\end{equation}
where the transition matrices are probability densities in $\Lambda$, that is a measure space with measure $d\lambda$, and the integrals are {Lebesgue integrals.} Moreover, generalized noncontextuality requirements (such as $P\simeq P' \implies \mu_P = \mu_{P'}$) are now defined up to sets of measure zero. Then, we define the set 
\begin{equation}
    \Bar{\Lambda}_k(\lambda)=\left \{ \lambda'|\xi(k|\lambda')>\xi(k|\lambda)\right \}, \label{Lambda_bar_def}
\end{equation}
 for each outcome $k$ and ontological state $\lambda$. Additionally, we have $\Bar{\Lambda}_k^c(\lambda) = \Lambda/\Bar{\Lambda}_k(\lambda)$. Now, Eq.~\eqref{op_variation_ont_model_infinite} implies 
\begin{align}
     \langle \Delta \mathcal{A}\rangle & = \sum_k a_k \Bigg( \int_{\Bar{\Lambda}_k(\lambda)} d\lambda' \int_\Lambda d\lambda \mu_P (\lambda) \Gamma_T (\lambda'|\lambda) \xi_{M}(k|\lambda') \nonumber \\
    & + \int_{\Bar{\Lambda}_k^c(\lambda)} d\lambda' \int_\Lambda d\lambda \mu_P (\lambda) \Gamma_T (\lambda'|\lambda) \xi_{M}(k|\lambda') - \int_\Lambda d\lambda \mu_P(\lambda) \xi_{M}(k | \lambda) \Bigg) \nonumber \\
    & \leq \sum_k a_k \Bigg( \int_{\Bar{\Lambda}_k(\lambda)} d\lambda' \int_\Lambda d\lambda \mu_P (\lambda) \Gamma_T (\lambda'|\lambda) \xi_{M}(k|\lambda') \nonumber \\
    & + \int_{\Bar{\Lambda}_k^c(\lambda)} d\lambda' \int_\Lambda d\lambda \mu_P (\lambda) \Gamma_T (\lambda'|\lambda) \xi_{M}(k|\lambda) - \int_\Lambda d\lambda \mu_P(\lambda) \xi_{M}(k | \lambda) \Bigg), 
\end{align}
where in the inequality we used that $\xi_{M}(k|\lambda') \leq \xi_{M}(k|\lambda)$ for $\lambda' \in \Bar{\Lambda}_k^c(\lambda)$. 

Therefore,
\begin{align}
     \langle \Delta \mathcal{A}\rangle & \leq \sum_k a_k \Bigg( \int_{\Bar{\Lambda}_k(\lambda)} d\lambda' \int_\Lambda d\lambda ~\mu_P (\lambda) \Gamma_T (\lambda'|\lambda) \xi_{M}(k|\lambda') \nonumber \\
    & + \int_{\Bar{\Lambda}_k^c(\lambda)} d\lambda' \int_\Lambda d\lambda ~\mu_P (\lambda) \Gamma_T (\lambda'|\lambda) \xi_{M}(k|\lambda) - \int_\Lambda d\lambda ~\mu_P(\lambda) \xi_{M}(k | \lambda) \Bigg) \nonumber \\
    & = \sum_k a_k \Bigg( \int_{\Bar{\Lambda}_k(\lambda)} d\lambda' \int_\Lambda d\lambda ~ \mu_P (\lambda) \Gamma_T (\lambda'|\lambda)(\xi_{M}(k|\lambda') -\xi_{M}(k|\lambda))  \nonumber \\
    & + \int_\Lambda d\lambda' \int_\Lambda d\lambda ~ \mu_P (\lambda) \Gamma_T (\lambda'|\lambda) \xi_{M}(k|\lambda) - \int_\Lambda d\lambda ~ \mu_P(\lambda) \xi_{M}(k | \lambda) \Bigg) \nonumber \\
    & = \sum_k a_k \left( \int_{\Bar{\Lambda}_k(\lambda)} d\lambda' \int_\Lambda d\lambda ~ \mu_P (\lambda) \Gamma_T \left(\lambda'|\lambda)(\xi_{M}(k|\lambda') -\xi_{M}(k|\lambda)\right)\right), \label{lemma2_continuous_step1}
\end{align}
 where the term $\sum_k a_ k\int_{\Bar{\Lambda}_k(\lambda)} d\lambda' \int_\Lambda d\lambda \mu_P (\lambda) \Gamma_T (\lambda'|\lambda) \xi_{M}(k|\lambda)$ was summed and subtracted in the first equality above, and in the last equality, we used that $\int_\Lambda d\lambda' \Gamma_T(\lambda'|\lambda)=1$.

 At this point, we make the \textit{transformation noncontextuality} assumption in the operational equivalences of Eqs.~\eqref{SR_condition_T2} and~\eqref{op_equiv_lemma2}. This results in equations identical to Eqs.~\eqref{noncotext_condition_lemma2} and~\eqref{noncotext_condition_lemma2_T2}, where the transition matrices are now probability densities, and where $\delta_{\lambda,\lambda'}$ is substituted by $\delta(\lambda-\lambda')$, which is a Dirac delta function. Using the combination of the continuous form of Eqs.~\eqref{noncotext_condition_lemma2} and~\eqref{noncotext_condition_lemma2_T2}, and the fact that $\Gamma_{A_1} (\lambda'|\lambda) \geq 0$, we have
\begin{align}
    \langle \Delta \mathcal{A}\rangle & \leq \sum_k a_k \Bigg( \int_{\Bar{\Lambda}_k(\lambda)} d\lambda' \int_\Lambda d\lambda \mu_P (\lambda) \Big(4(1-p_{d_1})(1-p_{d_1}) \delta(\lambda - \lambda') \nonumber \\
    & + 4(1-p_{d_1})p_{d_2} \Gamma_{C_2} (\lambda'|\lambda) + 2 p_{d_1} \Gamma_{A_2} (\lambda'|\lambda) \Big)(\xi_{M}(k|\lambda') -\xi_{M}(k|\lambda))\Bigg),
\end{align}

The term with the Dirac delta in the inequality above is null after performing the double integral, since the space $\Bar{\Lambda}_k(\lambda)$ has only terms with $\lambda' \neq \lambda$. Additionally, since $\xi_{M} (k|\lambda) \geq 0$, we have
\begin{align}
   \langle \Delta \mathcal{A}\rangle & \leq \sum_k a_k \Bigg( \int_{\Bar{\Lambda}_k(\lambda)} d\lambda' \int_\Lambda d\lambda ~ \mu_P (\lambda)\Big( 4(1-p_{d_1})p_{d_2} \Gamma_{C_2} (\lambda'|\lambda) + 2 p_{d_1} \Gamma_{A_2} (\lambda'|\lambda) \Big)\xi_{M}(k|\lambda') \Bigg) \nonumber \\
    & \leq \sum_k a_k \Bigg( \int_{\Lambda} d\lambda' \int_\Lambda d\lambda~ \mu_P (\lambda)\Big( 4(1-p_{d_1})p_{d_2} \Gamma_{C_2} (\lambda'|\lambda) + 2 p_{d_1} \Gamma_{A_2} (\lambda'|\lambda) \Big)\xi_{M}(k|\lambda') \Bigg),   \label{continuous_lambda_delta_A_step1}
\end{align}
where in the last inequality we summed the term $ \int_{\Bar{\Lambda}_k^c(\lambda)} d\lambda' \int_\Lambda d\lambda~\mu_P (\lambda)\Big( (1-p_d)\epsilon \Gamma'(\lambda'|\lambda) + p_d \Gamma_C(\lambda'|\lambda)\Big)\xi_{M}(k|\lambda'),$ which is non-negative. 

Noticing that $\int_\Lambda d\lambda' \int_\Lambda d\lambda \mu_P\Gamma_{C_2}(\lambda'|\lambda)\xi_{M}(k|\lambda') = p(k|M,T_{C_2},P) $ and $\int_\Lambda d\lambda' \int_\Lambda d\lambda \mu_P\Gamma_{A_2}(\lambda'|\lambda)\xi_{M}(k|\lambda') = p(k|M,T_{A_2},P)$ are probabilities obtaining outcome $k$, given a set of operational equivalences, we have
\begin{align}
 \sum_k a_k \int_\Lambda d\lambda' \int_\Lambda d\lambda \mu_P\Gamma_{C_2(A_2)}(\lambda'|\lambda)\xi_{M}(k|\lambda')  = \sum_k a_k p(k|M,T_{C_2(A_2)},P) \leq a_{max}. \label{probability_ineq_infinte}
\end{align}
Using the inequality~\eqref{continuous_lambda_delta_A_step1}, we obtain
\begin{equation}
    \langle \Delta \mathcal{A}\rangle  \leq (4(1-p_{d_1})p_{d_2} + 2 p_{d_1}) a_{max} = 2(p_{d_1}+ 2p_{d_2}- 2p_{d_1}p_{d_2}) a_{max}. \label{continuous_inequality_result1}
\end{equation}

To prove the converse, we start again from Eq.~\eqref{op_variation_ont_model_infinite}, which implies 
\begin{equation*}
    -\langle \Delta \mathcal{A}\rangle = - \langle \mathcal{A}(t) - \mathcal{A}(0) \rangle = \sum_k a_k \left(\int_\Lambda d\lambda \mu_P(\lambda) \xi_{M}(k | \lambda) - \int_\Lambda d\lambda \int_\Lambda d\lambda' \mu_P (\lambda) \Gamma_T (\lambda'|\lambda) \xi_{M}(k|\lambda') \right). 
\end{equation*}
Now we make the transformation noncontextuality assumption, using the infinite version of Eqs.~\eqref{noncotext_condition_lemma2} and~\eqref{noncotext_condition_lemma2_T2}. Hence, we obtain
\begin{align}
    - \langle \Delta \mathcal{A}\rangle \rangle & = \sum_k a_k \Bigg(\int_\Lambda d\lambda \mu_P(\lambda) \xi_{M}(k | \lambda) - \int_\Lambda d\lambda \int_\Lambda d\lambda' \mu_P (\lambda) \Big[ -2(1-p_{d_1})\Gamma_{T_2^*} (\lambda'|\lambda) - \Gamma_{A_1} (\lambda'|\lambda)  \nonumber \\
    & + 4(1-p_{d_1})(1-p_{d_2}) \delta(\lambda-\lambda') +4(1-p_{d_1})p_{d_2} \Gamma_{C_2} (\lambda'|\lambda)+ 2p_{d_1} \Gamma_{A_2} (\lambda'|\lambda) \Big] \xi_{M}(k|\lambda') \Bigg).
\end{align}
From the fact that $ -4(1-p_{d_1})p_{d_2}\int_\Lambda d\lambda \int_\Lambda d\lambda' \mu_P (\lambda) \Gamma_{C_2} \xi_{M}(k | \lambda')$ and 
$ -2p_{d_1}\int_\Lambda d\lambda \int_\Lambda d\lambda' \mu_P (\lambda) \Gamma_{A_2} (\lambda'|\lambda) \xi_{M}(k|\lambda')$ are never positive, we conclude that
\begin{align}
     - \langle \Delta \mathcal{A}\rangle & \leq \sum_k a_k \Bigg(\int_\Lambda d\lambda \mu_P(\lambda) \xi_{M}(k | \lambda) - \int_\Lambda d\lambda \int_\Lambda d\lambda' \mu_P (\lambda) \Big( -2(1-p_{d_1})\Gamma_{T_2^*} (\lambda'|\lambda) - \Gamma_{A_1} (\lambda'|\lambda)  \nonumber \\
    & + 4(1-p_{d_1})(1-p_{d_2}) \delta(\lambda-\lambda') \Big) \Bigg) \nonumber \\
    & = ~ \sum_k a_k \Bigg[  2(p_{d_1}+2p_{d_2}-2p_{d_1}p_{d_2})  \int_\Lambda d\lambda \mu_P (\lambda) \xi_{M}(k|\lambda) \nonumber \\
   & + 2(1-p_{d_1})\left( \int_\Lambda d\lambda \int_\Lambda d\lambda'  \mu_P (\lambda)\Gamma_{T_2^*} (\lambda'|\lambda)\xi_{M}(k|\lambda') - \int_\Lambda d\lambda \mu_P (\lambda) \xi_{M}(k|\lambda) \right) \nonumber \\
   & + \left( \int_\Lambda d\lambda \int_\Lambda d\lambda'  \mu_P (\lambda)\Gamma_{A_1} (\lambda'|\lambda)\xi_{M}(k|\lambda')  - \int_\Lambda d\lambda \mu_P (\lambda) \xi_{M}(k|\lambda) \right) \Bigg],\label{inequality_continuous_big_factorization}
\end{align}
where in the last step we only evaluated the integral over the $\delta(\lambda-\lambda')$ and reorganized the terms. 

Moreover, the term $\sum_k a_k \Big(\int_\Lambda d\lambda \int_\Lambda d\lambda' \mu_P (\lambda) \Gamma_{T_{A_1}} (\lambda'|\lambda) \xi_{M}(k | \lambda') -\int_\Lambda d\lambda \mu_P(\lambda) \xi_{M}(k | \lambda) \Big)$ is analogous to the r.h.s. of Eq.~\eqref{op_variation_ont_model_infinite} with the exchange of $\Gamma_T(\lambda'|\lambda)$ by $\Gamma_{T_{A_1}}(\lambda'|\lambda)$. Since these terms are symmetric in the noncontextuality condition (the continuous version of Eq.~\eqref{noncotext_condition_lemma2}), we can repeat the same arguments from Eq.~\eqref{op_variation_ont_model_infinite} to the inequality~\eqref{continuous_inequality_result1} and conclude that 
\begin{equation}
    \sum_k a_k \Bigg(\int_\Lambda d\lambda \int_\Lambda d\lambda' \mu_P (\lambda) \Gamma_{A_1} (\lambda'|\lambda) \xi_{M}(k | \lambda') -\int_\Lambda d\lambda \mu_P(\lambda) \xi_{M}(k | \lambda) \Bigg) \leq 2(p_{d_1} + 2p_{d_2} - 2p_{d_1}p_{d_2}) a_{max}. \label{lemma2_almost_last_setp}
\end{equation}
Similarly, the term $\sum_k a_k \Big(\int_\Lambda d\lambda \int_\Lambda d\lambda' \mu_P (\lambda) \Gamma_{T_2^*} (\lambda'|\lambda) \xi_{M}(k | \lambda') -\int_\Lambda d\lambda \mu_P(\lambda) \xi_{M}(k | \lambda) \Big)$ is also analogous to the r.h.s. of Eq.~\eqref{op_variation_ont_model_infinite} with the exchange of $\Gamma_T(\lambda'|\lambda)$ by $\Gamma_{T_2^*}(\lambda'|\lambda)$. Therefore, we repeat the same arguments from Eq.~\eqref{op_variation_ont_model_infinite} to the inequality of~\eqref{lemma2_continuous_step1} to obtain
\begin{align}
   & \sum_k a_k \Bigg(\int_\Lambda d\lambda \int_\Lambda d\lambda' \mu_P (\lambda) \Gamma_{T_2^*} (\lambda'|\lambda) \xi_{M}(k | \lambda') -\int_\Lambda d\lambda \mu_P(\lambda) \xi_{M}(k | \lambda) \Bigg) \nonumber \\ &  \leq \sum_k a_k \left( \int_{\Bar{\Lambda}_k(\lambda)} d\lambda' \int_\Lambda d\lambda ~ \mu_P (\lambda) \Gamma_{T_2^*} \left(\lambda'|\lambda)(\xi_{M}(k|\lambda') -\xi_{M}(k|\lambda)\right)\right).
\end{align}
Now we use the continuous version of the noncontextuality condition of Eq.~\eqref{noncotext_condition_lemma2_T2} and use the fact that $\Gamma_{T_2}(\lambda'|\lambda) \geq 0$ to obtain
\begin{align}
   & \sum_k a_k \Bigg(\int_\Lambda d\lambda \int_\Lambda d\lambda' \mu_P (\lambda) \Gamma_{T_2^*} (\lambda'|\lambda) \xi_{M}(k | \lambda') -\int_\Lambda d\lambda \mu_P(\lambda) \xi_{M}(k | \lambda) \Bigg) \nonumber \\ &  \leq \sum_k a_k \left( \int_{\Bar{\Lambda}_k(\lambda)} d\lambda' \int_\Lambda d\lambda ~ \mu_P (\lambda) \left(2(1- p_{d_2}) \delta(\lambda-\lambda') + 2p_{d_2} \Gamma_{C_2}(\lambda'|\lambda) \right) (\xi_{M}(k|\lambda') -\xi_{M}(k|\lambda))\right).
\end{align}
Again, the double integral over the $\delta(\lambda-\lambda')$ will be null due to the integration over the space $\Bar{\Lambda}_k(\lambda)$, and we use that $\xi_{M}(k|\lambda) \geq 0$ to obtain
\begin{align}
   & \sum_k a_k \Bigg(\int_\Lambda d\lambda \int_\Lambda d\lambda' \mu_P (\lambda) \Gamma_{T_2^*} (\lambda'|\lambda) \xi_{M}(k | \lambda') -\int_\Lambda d\lambda \mu_P(\lambda) \xi_{M}(k | \lambda) \Bigg) \nonumber \\ &  \leq 2p_{d_2}\sum_k a_k \left( \int_{\Bar{\Lambda}_k(\lambda)} d\lambda' \int_\Lambda d\lambda ~ \mu_P (\lambda) \Gamma_{C_2}(\lambda'|\lambda) \xi_{M}(k|\lambda')\right) \nonumber \\
   &  \leq 2p_{d_2}\sum_k a_k \left( \int_\Lambda d\lambda' \int_\Lambda d\lambda ~ \mu_P (\lambda) \Gamma_{C_2}(\lambda'|\lambda) \xi_{M}(k|\lambda')\right),
\end{align}
where in the last inequality we used that fact that $\int_{\Bar{\Lambda}_k(\lambda)} d\lambda' \int_\Lambda d\lambda ~ \mu_P (\lambda) \Gamma_{C_2}(\lambda'|\lambda) \xi_{M}(k|\lambda') \leq \int_\Lambda d\lambda' \int_\Lambda d\lambda ~ \mu_P (\lambda) \Gamma_{C_2}(\lambda'|\lambda) \xi_{M}(k|\lambda')$. Now, again we notice that $\int_\Lambda d\lambda' \int_\Lambda d\lambda ~ \mu_P (\lambda) \Gamma_{C_2}(\lambda'|\lambda) \xi_{M}(k|\lambda') = p(k| M,T_{C_2},P)$ is a probability of obtaining the outcome $k$. Hence, this quantity must satisfy an inequality similar to~\eqref{probability_ineq_infinte}, from which we obtain
\begin{align}
    \sum_k a_k \Bigg(\int_\Lambda d\lambda \int_\Lambda d\lambda' \mu_P (\lambda) \Gamma_{T_2^*} (\lambda'|\lambda) \xi_{M}(k | \lambda') -\int_\Lambda d\lambda \mu_P(\lambda) \xi_{M}(k | \lambda) \Bigg)  \leq 2 p_{d_2} a_{max}. \label{lemma2_last_step}
\end{align}

Finally, since $\int_\Lambda d\lambda \mu_P(\lambda) \xi_{M}(k | \lambda) = p(k|M,P)$ is the probability distribution for obtaining the outcome $k$ in the prepare and measure case $(P,M)$, it must also satisfy $$\sum_k a_k \int_\Lambda d\lambda \mu_P(\lambda) \xi_{M}(k | \lambda) = \sum_k a_k p(k|M,P) \leq a_{max}.$$ Using the inequality above, and the inequalities~\eqref{lemma2_almost_last_setp} and~\eqref{lemma2_last_step} in~\eqref{inequality_continuous_big_factorization}, we obtain
\begin{equation}
     - \langle \Delta \mathcal{A}\rangle \leq ~ 2(p_{d_1}+2p_{d_2}-2p_{d_1}p_{d_2}) a_{max} +  4(1-p_{d_1})p_{d_2} a_{max} + 3(p_{d_1}+2p_{d_2}-2p_{d_1}p_{d_2})a_{max} \leq 4( p_{d_1} + 3 p_{d_2} - 3 p_{d_1} p_{d_2}), \nonumber
\end{equation}
which completes the proof for the inequality for infinite $\Lambda$.

\section{Proof of Theorem \ref{theorem_3}}\label{Appendix_Theorem_Quantum}


We start showing that a general interaction Hamiltonian between two \textit{non-resonant} qubits, which conserves the total energy of the qubits, must be of the following form
\begin{equation}
    H_{I_\text{nr}}= g \begin{pmatrix}
        0 & 0 & 0 & 0 \\
        0 & 1 & 0 & 0 \\
        0 & 0 & 1 & 0 \\
        0 & 0 & 0 & 0
    \end{pmatrix}, \label{general_hi_non_resonant2}
\end{equation}
for a real number, $g$. In the case of an  interaction Hamiltonian between two \textit{resonant} qubits, it must have the following form
\begin{equation}
    H_{I_\text{r}}= g \begin{pmatrix}
        0 & 0 & 0 & 0 \\
        0 & a & e^{i \theta} & 0 \\
        0 & e^{-i \theta} & a & 0 \\
        0 & 0 & 0 & 0
    \end{pmatrix}, \label{general_hi_resonant2}
\end{equation}
where $a$, $g$ and $\theta$ are real parameters.

Note that the most general interaction Hamiltonian between two qubits is the most general linear combination of tensor products of Pauli matrices acting on the Hilbert spaces $\mathcal{H}_A = \mathbb{C}^2$ and $\mathcal{H}_B = \mathbb{C}^2$
\begin{align}
    H_{I} = & ~ \alpha_0 \mathbb{1}^A \otimes \mathbb{1}^B + \alpha_{11} \sigma_x^A \otimes \sigma_x^B + \alpha_{12} \sigma_x^A \otimes \sigma_y^B + \alpha_{13} \sigma_x^A \otimes \sigma_z^B \nonumber \\
    & + \alpha_{21} \sigma_y^A \otimes \sigma_x^B + \alpha_{22} \sigma_y^A \otimes \sigma_y^B + \alpha_{23} \sigma_y^A \otimes \sigma_z^B \nonumber \\
    & + \alpha_{31} \sigma_z^A \otimes \sigma_x^B + \alpha_{32} \sigma_z^A \otimes \sigma_y^B + \alpha_{33} \sigma_z^A \otimes \sigma_z^B, \label{most_general_Hi}
\end{align}
where all coefficients must be real numbers.

Furthermore, the assumption that the interaction Hamiltonian conserves the sum of the two local energies is equivalent to 
\begin{equation}
    [H_{I}, H_A \otimes \mathbb{1}_B + \mathbb{1}_A \otimes H_B] = 0, \label{energy_conservation2}
\end{equation}
where $H_A$ is the local Hamiltonian of the qubit $A$ and $H_B$ is the local Hamiltonian of the qubit $B$.

We suppose the qubits have local Hamiltonians $H_A,H_B$ described by what we refer to as \emph{Zeeman Hamiltonians}. If the two qubits are \textit{non-resonant}, without loss of generality, we say that $H_A= \frac{\omega_A}{2}(c_1 \mathbb{1}^A+c_2 \sigma_z^A)$ and $H_B =\frac{\omega_B}{2}(c_3 \mathbb{1}^B+c_4 \sigma_z^B)$, for $\omega_A \neq \omega_B$ and $c_i \in \mathbb{R}$. With these assumptions, Eq.~\eqref{energy_conservation2} results in a set of $16$ equations, whose solutions imply in the following interaction Hamiltonian 
\begin{equation*}
    H_{I_\text{nr}} = \begin{pmatrix}
        \alpha_0 + \alpha_{33} & 0 & 0 & 0 \\
        0 & \alpha_0 - \alpha_{33} & 0 & 0 \\
        0 & 0 & \alpha_0 - \alpha_{33} & 0 \\
        0 & 0 & 0 & \alpha_0 + \alpha_{33}
    \end{pmatrix}.
\end{equation*}
Since the sum of identity terms on the Hamiltonian does not change the dynamics, we are free to select the value of $\alpha_0$. We choose $\alpha_0 = -\alpha_{33}$ {and define $g  = - 2 \alpha_{33}$}, which recovers Eq.~\eqref{general_hi_non_resonant2}.

Similarly, if the qubits are \textit{resonant}, we suppose $H_A = H_B = \frac{\omega}{2}(c_1 \mathbb{1}^{A(B)}+c_2 \sigma_z^{A(B)})$. For this case, Eq.~\eqref{energy_conservation2} results in a set of $16$ equations whose solutions imply the following interaction Hamiltonian
\begin{equation*}
    H_{I_\text{r}} = \begin{pmatrix}
        \alpha_0 + \alpha_{33} & 0 & 0 & 0 \\
        0 & \alpha_0 - \alpha_{33} & 2\alpha_{22} - 2 i\alpha_{21} & 0 \\
        0 & 2\alpha_{22} + 2 i\alpha_{21} & \alpha_0 - \alpha_{33} & 0 \\
        0 & 0 & 0 & \alpha_0 + \alpha_{33}
    \end{pmatrix}.
\end{equation*}
Again, we select $\alpha_0 = -\alpha_{33}$ and define $a, g$, and $\theta$, such that $g a = - 2 \alpha_{33}$, and $g e^{-i \theta} = 2\alpha_{22} + 2 i\alpha_{21}$, which imply Eq.~\eqref{general_hi_resonant2}. 


To prove the main result of the theorem, we start with the {\textit{non-resonant}} qubits case. We desire to prove that the unitary $U_{I_\text{nr}} = e^{-i H_{I_\text{nr}}}$ satisfies an equation in the form of Eq.~\eqref{eq:theorem_nonresonant}. To do this, we write the most general form of a two qubits density matrix
\begin{equation}
    \rho_{\text{gen}} = \begin{pmatrix}
        p_{00} & \nu_1^* & \nu_2^* & \gamma^* \\
        \nu_1  &  p_{01} & \eta e^{i \xi} & \nu_3^* \\
        \nu_2 & \eta e^{-i \xi} &  p_{10} & \nu_4^* \\
        \gamma & \nu_3  & \nu_4  &  p_{11}
    \end{pmatrix}, \label{2_qubit_general_rho}
\end{equation} with $p_{00}+p_{01}+p_{10}+p_{11}=1$, where diagonal terms are real non-negative numbers while  $\nu_1,~\nu_2,~\nu_3,~\nu_4,~\gamma$ are complex numbers, and $\eta,~\xi$ are real; all parameters are constrained such that $\rho_{\text{gen}} \geq 0$. And now, given the unitary evolution, $U_\text{nr} = e^{-i t H_\text{nr}}$, generated by the interaction Hamiltonian given by Eq.~\eqref{general_hi_non_resonant2}, we have 
\begin{align}
    & \frac{1}{2} \mathcal{U}_\text{nr}[\rho_\text{gen}]+\frac{1}{2} \mathcal{U}^\dagger_\text{nr}[\rho_\text{gen}]  = \frac{1}{2}U_\text{nr} \rho_{\text{gen}} U_\text{nr}^\dagger + \frac{1}{2}U_\text{nr}^\dagger \rho_{\text{gen}} U_\text{nr} = \begin{pmatrix}
        p_{00} & \nu_1^*\cos(gt) & \nu_2^*\cos(gt) & \gamma^* \\
        \nu_1\cos(gt)  &  p_{01} & \eta e^{i \xi} & \nu_3^*\cos(gt) \\
        \nu_2\cos(gt) & \eta e^{-i \xi} &  p_{10} & \nu_4^*\cos(gt) \\
        \gamma & \nu_3\cos(gt)  & \nu_4\cos(gt)  &  p_{11}
    \end{pmatrix}. 
\end{align}
This can be factorized as 
\begin{align}
      \frac{1}{2} \mathcal{U}_\text{nr}[\rho_\text{gen}]+\frac{1}{2} \mathcal{U}^\dagger_\text{nr}[\rho_\text{gen}] = (1-p_{d_\text{nr}}) \rho_{\text{gen}} + p_{d_\text{nr}} \mathcal{C}_\text{nr}[\rho_\text{gen}],  \label{SR_condition_proof_nr}
\end{align}
where $p_{d_\text{nr}} = \sin^2(gt/2)$, and 
\begin{equation}
    \mathcal{C}_\text{nr}[\rho_\text{gen}] := \begin{pmatrix}
        p_{00} & -\nu_1^* & -\nu_2^* & \gamma^* \\
        -\nu_1  &  p_{01} & \eta e^{i \xi} & -\nu_3^* \\
        -\nu_2 & \eta e^{-i \xi} &  p_{10} & -\nu_4^* \\
        \gamma & -\nu_3  & -\nu_4  &  p_{11} \label{C_nr_map_def}
    \end{pmatrix}.
\end{equation}
To prove that Eq.~\eqref{SR_condition_proof_nr} is really as Eq.~\eqref{eq:theorem_nonresonant}, we must prove that the channel $\mathcal{C}_\text{nr}(\bullet)$ defined above is a completely positive and trace-preserving map. To verify  \textit{complete positivity}, we can compute its Choi-Jamiolkowski matrix~\cite{nielsen2010quantum,wilde2013quantum} and prove that this matrix is positive semi-definite. The Choi-Jamiolkowski matrix for a map $\mathcal{C}(\bullet)$ is defined as the following
\begin{equation}
    \Lambda_{\mathcal{C}} := \mathcal{I}^\mathcal{R} \otimes \mathcal{C}^\mathcal{S} \left( \ket{\Tilde{\Psi}}\bra{\Tilde{\Psi}}\right) = \sum_{i=0,j=0}^{3,3} \ket{i}\bra{j}^\mathcal{R} \otimes \mathcal{C} \left( \ket{i}\bra{j} \right)^\mathcal{S} , \label{CJ_matrix_def}
\end{equation}
where the index $\mathcal{R}$ stands for the channel acting on an auxiliary Hilbert space $\mathcal{H}_\mathcal{R}$ with dimension $4$, the index $\mathcal{S}$ stands for the channel acting on the global two-qubits system Hilbert space $\mathcal{H}_\mathcal{S}$, $ \mathcal{I}$ is the identity channel, and $\ket{\Tilde{\Psi}} = \sum_{j=0}^3 \ket{j}^\mathcal{R} \otimes \ket{j}^\mathcal{S}$ is the unnormalized maximally entangled state in $\mathcal{H}_\mathcal{R} \otimes \mathcal{H}_\mathcal{S}$.

Computing explicitly the Choi-Jamiolkowski matrix $\Lambda_{\mathcal{C}_\text{nr}}$ for the map $\mathcal{C}_\text{nr}(\bullet)$, we obtain 
\begin{equation*}
    \Lambda_{\mathcal{C}_\text{nr}} =\left(
\begin{array}{cccccccccccccccc}
 1 & 0 & 0 & 0 & 0 & 0 & -1 & 0 & 0 & -1 & 0 & 0 & 0 & 0 & 0 & 1 \\
 0 & 0 & 0 & 0 & 0 & 0 & 0 & 0 & 0 & 0 & 0 & 0 & 0 & 0 & 0 & 0 \\
 0 & 0 & 0 & 0 & 0 & 0 & 0 & 0 & 0 & 0 & 0 & 0 & 0 & 0 & 0 & 0 \\
 0 & 0 & 0 & 0 & 0 & 0 & 0 & 0 & 0 & 0 & 0 & 0 & 0 & 0 & 0 & 0 \\
 0 & 0 & 0 & 0 & 0 & 0 & 0 & 0 & 0 & 0 & 0 & 0 & 0 & 0 & 0 & 0 \\
 0 & 0 & 0 & 0 & 0 & 0 & 0 & 0 & 0 & 0 & 0 & 0 & 0 & 0 & 0 & 0 \\
 -1 & 0 & 0 & 0 & 0 & 0 & 1 & 0 & 0 & 1 & 0 & 0 & 0 & 0 & 0 & -1 \\
 0 & 0 & 0 & 0 & 0 & 0 & 0 & 0 & 0 & 0 & 0 & 0 & 0 & 0 & 0 & 0 \\
 0 & 0 & 0 & 0 & 0 & 0 & 0 & 0 & 0 & 0 & 0 & 0 & 0 & 0 & 0 & 0 \\
 -1 & 0 & 0 & 0 & 0 & 0 & 1 & 0 & 0 & 1 & 0 & 0 & 0 & 0 & 0 & -1 \\
 0 & 0 & 0 & 0 & 0 & 0 & 0 & 0 & 0 & 0 & 0 & 0 & 0 & 0 & 0 & 0 \\
 0 & 0 & 0 & 0 & 0 & 0 & 0 & 0 & 0 & 0 & 0 & 0 & 0 & 0 & 0 & 0 \\
 0 & 0 & 0 & 0 & 0 & 0 & 0 & 0 & 0 & 0 & 0 & 0 & 0 & 0 & 0 & 0 \\
 0 & 0 & 0 & 0 & 0 & 0 & 0 & 0 & 0 & 0 & 0 & 0 & 0 & 0 & 0 & 0 \\
 0 & 0 & 0 & 0 & 0 & 0 & 0 & 0 & 0 & 0 & 0 & 0 & 0 & 0 & 0 & 0 \\
 1 & 0 & 0 & 0 & 0 & 0 & -1 & 0 & 0 & -1 & 0 & 0 & 0 & 0 & 0 & 1 \\
\end{array}
\right).
\end{equation*}
This matrix has eigenvalues $\{ 0, 0, 0, 0, 0, 0, 0, 0, 0, 0, 0, 0, 0, 0, 0, 4 \}$, therefore, it is positive semi-definite. Hence, the channel $\mathcal{C}_\text{nr}(\bullet)$ is completely positive. Since the trace of the map described in Eq.~\eqref{C_nr_map_def} is the same as the initial density matrix in Eq.~\eqref{2_qubit_general_rho}, the channel is also clearly \textit{trace-preserving}. Therefore, Eq.~\eqref{eq:theorem_nonresonant} holds, as we wanted to show. 


To prove the resonant  case, we start by noticing that the interaction Hamiltonian of Eq.~\eqref{general_hi_resonant2} can be written in the following form
\begin{equation*}
    H_{I_\textbf{r}} = H_\theta + H_a,
\end{equation*}
where
\begin{equation*}
    H_\theta = g \begin{pmatrix}
        0 & 0 & 0 & 0 \\
        0 & 1 & e^{i \theta} & 0 \\
        0 & e^{-i \theta} & 1 & 0 \\
        0 & 0 & 0 & 0
    \end{pmatrix},
\end{equation*}
and
\begin{equation*}
    H_a= g \begin{pmatrix}
        0 & 0 & 0 & 0 \\
        0 & a-1 & 0 & 0 \\
        0 & 0 & a-1 & 0 \\
        0 & 0 & 0 & 0
    \end{pmatrix}.
\end{equation*}

Since $U_\text{r} = e^{-i t H_{I_\text{r}}}$,  because $[H_\theta, H_a]=0$ we have that $U_\text{r}$ can be decomposed as 
\begin{equation*}
    U_\text{r} = U_\theta U_a,
\end{equation*}
where $U_\theta = e^{-i t H_\theta}$ and $U_a = e^{-i t H_a}$. For the rest of the proof, we show that $U_a$ and $U_\theta$ each satisfy  Eqs.~\eqref{theorem_3_resonant_eq1} and~\eqref{theorem_3_resonant_eq2}, respectively.

Again, we suppose the initial global state starts in its general form Eq.~\eqref{2_qubit_general_rho}. For the case of $U_a$, notice that $H_a$ has the same form as $H_{I_\text{nr}}$ (Eq.~\eqref{general_hi_non_resonant2}) with the exchange of $g$ by $g(a-1)$. Therefore, a similar equation to Eq.~\eqref{SR_condition_proof_nr} is immediately valid, namely 
\begin{equation*}
    \frac{1}{2} \mathcal{U}_\theta [\rho_\text{gen}]+\frac{1}{2} \mathcal{U}^\dagger_\theta[\rho_\text{gen}] = (1-p_1) \rho_{\text{gen}} + p_1 \mathcal{C}_1[\rho_\text{gen}], 
\end{equation*}
where $p_1= \sin^2((a-1)gt/2)$, and $\mathcal{C}_1[\rho_\text{gen}]$ is defined exactly as in Eq.~\eqref{C_nr_map_def}, which we already proved is a CPTP map.

As for $U_\theta$, we compute
\begin{align}
    & \frac{1}{2} \mathcal{U}_\theta[\rho_\text{gen}]+\frac{1}{2} \mathcal{U}^\dagger_\theta [\rho_\text{gen}] = \frac{1}{2}U_\theta \rho_{\text{gen}} U_\theta^\dagger + \frac{1}{2}U_\theta^\dagger \rho_{\text{gen}} U_\theta = \begin{pmatrix}
        p_{00} &  f_\theta(2,1)^* & f_\theta(3,1)^* & \gamma^* \\
        f_\theta(2,1) &  f_\theta(2,2) & f_\theta(3,2)^* & f_\theta(4,2)^* \\
        f_\theta(3,1) & f_\theta(3,2) &  f_\theta(3,3) & f_\theta(4,3)^* \\
       \gamma & f_\theta(4,2)  & f_\theta(4,3) &  p_{11}
    \end{pmatrix}, \label{SR_general_qubits_step1}
\end{align}
where now we have a more complicated matrix defined via the functions
\begin{align}
    f_\theta(2,2) & =\frac{1}{2} \left((p_{01}-p_{10}) \cos (2 g  t)+p_{01}+p_{10} \right), \nonumber \\
    f_\theta(3,3) & =\frac{1}{2} \left(p_{01}+p_{10}+(p_{10}-p_{01}) \cos(2 g t) \right), \nonumber \\
    f_\theta(2,1) & = \nu_1 \cos^2 (g t)-\nu_2 e^{i \theta} \sin^2 (g  t),\nonumber \\
    f_\theta(3,1) & = \nu_2 \cos^2 (g  t)-\nu_1 e^{-i \theta } \sin^2 (g  t),\nonumber \\
    f_\theta(3,2) & = \eta e^{-i \theta } (\cos (\xi-\theta)-i \cos (2 g  t) \sin (\xi -\theta )), \nonumber \\
    f_\theta(4,2) & = \nu_3 \cos^2 (g  t)- \nu_4 e^{-i \theta } \sin^2 (g  t),\nonumber \\
    \text{and} & \nonumber \\
    f_\theta(4,3) & =\nu_4 \cos^2 (g  t)-\nu_3 e^{i \theta } \sin^2 (g  t). \nonumber
\end{align}
Eq.~\eqref{SR_general_qubits_step1} can be factorized as 
\begin{align}
     \frac{1}{2} \mathcal{U}_\theta[\rho_\text{gen}]+\frac{1}{2} \mathcal{U}^\dagger_\theta [\rho_\text{gen}] = (1-p_{d_2}) \rho_{\text{gen}} + p_{d_2} \mathcal{C}_2[\rho_\text{gen}], \label{SR_qubits_theta}
\end{align}
where
\begin{equation*}
    p_{d_2} = \sin^2(gt), 
\end{equation*}
and 
\begin{align}
    \mathcal{C}_2[\rho_\text{gen}] = \begin{pmatrix}
        p_{00} &  -\nu_2^* e^{-i \theta }  &  -\nu_1^* e^{i \theta} & \gamma^* \\
        -\nu_2 e^{i \theta }  &  p_{10} & \eta e^{-i \xi}e^{i 2\theta} & -\nu_4^* e^{i \theta} \\
        -\nu_1 e^{-i \theta} & \eta e^{i \xi} e^{-i 2\theta} &  p_{01} & -\nu_3^* e^{-i \theta} \\
        \gamma & -\nu_4 e^{-i \theta}  & -\nu_3 e^{i \theta} &  p_{11}
    \end{pmatrix}. \label{C_definition}
\end{align}

Eq.~\eqref{SR_qubits_theta} gives us the desired form of factorization (i.e., the form of Eq.~\eqref{SR_condition_T2}). The remaining step is to prove that the map $\mathcal{C}_2(\bullet)$, described in~\eqref{C_definition}, is a completely positive trace-preserving map. Again, trace preservation follows trivially, so we proceed to show complete positivity.

Computing explicitly $\Lambda_{\mathcal{C}_2}$ (defined in Eq.~\eqref{CJ_matrix_def}), we obtain
\begin{equation*}
    \Lambda_{\mathcal{C}_2} =\left(
\begin{array}{cccccccccccccccc}
 1 & 0 & 0 & 0 & 0 & 0 & -e^{i \theta } & 0 & 0 & -e^{-i \theta } & 0 & 0 & 0 & 0 & 0 & 1 \\
 0 & 0 & 0 & 0 & 0 & 0 & 0 & 0 & 0 & 0 & 0 & 0 & 0 & 0 & 0 & 0 \\
 0 & 0 & 0 & 0 & 0 & 0 & 0 & 0 & 0 & 0 & 0 & 0 & 0 & 0 & 0 & 0 \\
 0 & 0 & 0 & 0 & 0 & 0 & 0 & 0 & 0 & 0 & 0 & 0 & 0 & 0 & 0 & 0 \\
 0 & 0 & 0 & 0 & 0 & 0 & 0 & 0 & 0 & 0 & 0 & 0 & 0 & 0 & 0 & 0 \\
 0 & 0 & 0 & 0 & 0 & 0 & 0 & 0 & 0 & 0 & 0 & 0 & 0 & 0 & 0 & 0 \\
 -e^{-i \theta } & 0 & 0 & 0 & 0 & 0 & 1 & 0 & 0 & e^{-2 i \theta } & 0 & 0 & 0 & 0 & 0 & -e^{-i \theta } \\
 0 & 0 & 0 & 0 & 0 & 0 & 0 & 0 & 0 & 0 & 0 & 0 & 0 & 0 & 0 & 0 \\
 0 & 0 & 0 & 0 & 0 & 0 & 0 & 0 & 0 & 0 & 0 & 0 & 0 & 0 & 0 & 0 \\
 -e^{i \theta } & 0 & 0 & 0 & 0 & 0 & e^{2 i \theta } & 0 & 0 & 1 & 0 & 0 & 0 & 0 & 0 & -e^{i \theta } \\
 0 & 0 & 0 & 0 & 0 & 0 & 0 & 0 & 0 & 0 & 0 & 0 & 0 & 0 & 0 & 0 \\
 0 & 0 & 0 & 0 & 0 & 0 & 0 & 0 & 0 & 0 & 0 & 0 & 0 & 0 & 0 & 0 \\
 0 & 0 & 0 & 0 & 0 & 0 & 0 & 0 & 0 & 0 & 0 & 0 & 0 & 0 & 0 & 0 \\
 0 & 0 & 0 & 0 & 0 & 0 & 0 & 0 & 0 & 0 & 0 & 0 & 0 & 0 & 0 & 0 \\
 0 & 0 & 0 & 0 & 0 & 0 & 0 & 0 & 0 & 0 & 0 & 0 & 0 & 0 & 0 & 0 \\
 1 & 0 & 0 & 0 & 0 & 0 & -e^{i \theta } & 0 & 0 & -e^{-i \theta } & 0 & 0 & 0 & 0 & 0 & 1 \\
\end{array}
\right).
\end{equation*}
This matrix has eigenvalues $\{ 0, 0, 0, 0, 0, 0, 0, 0, 0, 0, 0, 0, 0, 0, 0, 4 \}$. Therefore, it is positive semi-definite and the channel $\mathcal{C}_2(\bullet)$ is completely positive, and this concludes the proof. 

{
\section{$\Delta \mathcal{I}(A:B)$ for qubits and the Clausius Inequality}
\label{qubits_mutual_information}
The relation between heat flow and the change in mutual information comes from    
\begin{equation}
    \Delta \mathcal{I}(A:B) = (\beta_A - \beta_B) \langle \mathcal{Q}_A \rangle - S(\rho_A'|| \rho_A) - S(\rho_B'||\rho_B), \label{mutual information change}
\end{equation}
which is a consequence of Eq.~\eqref{second_law_gen_deduction}, derived in Appendix \ref{Appendix_A}.
To show that, for short times, the main term contributing to an anomalous heat flow is also the one which contributes to a negative mutual information difference $\Delta \mathcal{I}$, we proceed as follows. Let the initial global two-qubit density matrix be given by Eq.~\eqref{general_rho_2qubits_thermal} (with $\nu_0=\nu_1=\nu_2=0$ for simplicity, since they are irrelevant for heat exchange). In this case, if the global state evolves according to the general resonant energy-conserving Hamiltonian (Eq.~\eqref{qubit_general_energy_conserving_Hamiltonian}), then for short times we obtain the following expression the change in mutual information:
\begin{align}
   & \Delta\mathcal{I}(A:B) = (\beta_A - \beta_B) 2 \eta \omega \sin(\xi - \theta) g t \nonumber \\
    & +  \left(\omega  (\beta_A - \beta_B) \left(\frac{1}{e^{\beta_B \omega }+1}-\frac{1}{e^{\beta_A \omega }+1}\right)-4 \eta ^2 \sin ^2(\theta -\xi) (\cosh (\omega \beta_A)+\cosh (\omega \beta_B)+2)\right) g^2 t^2  + \mathcal{O}(g^3 t^3).
\end{align}
The full analytic expression for the mutual information is impractical to display here, but for our analysis, it is sufficient to consider the first two orders in $gt$. 

The first-order term, $(\beta_A - \beta_B) 2 \eta \omega \sin(\xi - \theta) g t$, is the leading contribution responsible for the negativity of $\Delta \mathcal{I}$, and is exactly $(\beta_A - \beta_B)$ times the first-order term of the heat average shown in Eq.~\eqref{general_heat_2qubits_thermal}. This comes from the heat contribution to the mutual information change, as given in Eq.~\eqref{mutual information change}, and is of the same order $gt$ heat flow term responsible for violating the noncontextuality inequality. Therefore, the same term that signals quantum contextuality and the reversal of heat flow is the one that contributes, in leading order, to the negative change in the mutual information at short times. 
}
{
\section{Proof of Eq. \eqref{heat_qudit_PSWAP}} \label{Appendix_qudits}
As described in the main text, we compute the heat received by a system $A$ which is a qudit with Hilbert space dimension $d \geq 2$. The local Hamiltonian of the qudit $A$ is given by Eq. \eqref{qudit_A_local_hamiltonian}. This qudit interacts with another qudit $B$, with Hilbert space dimension $d$, via a Partial SWAP unitary, given by \eqref{u to general swap}. Let the initial density matrix of the joint system $AB$ be given by Eq. \eqref{initial_state_qudit_AB}, then the heat received by $A$, after it interacts during an interval of time $t$ with $B$ is 
\begin{equation}
    \langle \mathcal{Q}_A \rangle = \tr \{ H_A(U_{PS} \rho U_{PS}^\dagger - \rho) \}. \label{qudit_heat_start}
\end{equation}

To compute this quantity, we make the following computation, using Eq. \eqref{u to general swap} 
\begin{align}
    U_{PS} \rho U_{PS}^\dagger = \cos^2(gt) \rho + i \sin (gt) \cos(gt) \rho S - i \sin(gt)\cos(gt) S \rho + \sin^2(gt) S \rho S.
\end{align}
This implies that
\begin{align}
    U_{PS} \rho U_{PS}^\dagger - \rho = \sin^2(gt) \left( S \rho S - \rho \right) + i \sin(gt) \cos(gt) \left( \rho S - S \rho \right).
\end{align}
Using this and Eqs. \eqref{qudit_A_local_hamiltonian} and \eqref{initial_state_qudit_AB} in Eq. \eqref{qudit_heat_start}, we obtain
\begin{align}
    \langle \mathcal{Q}_A \rangle & = \tr \Bigg( \sum_{k=0}^{d-1} \omega_k \ket{k}_A \bra{k}_A \otimes \mathbb{1}_B \Big[ \sin^2(gt)  \sum_{n,m,n',m'} p_{n,m,n',m'} ( S \ket{n,m}\bra{n',m'} S - \ket{n,m}\bra{n',m'}) \nonumber \\
    & + i \sin(gt) \cos(gt) \sum_{n,m,n',m'} p_{n,m,n',m'} (\ket{n,m}\bra{n',m'}S - S\ket{n,m}\bra{n',m'}) \Big] \Bigg) \nonumber \\
     & = \tr \Bigg( \sum_{k=0}^{d-1} \omega_k \ket{k}_A \bra{k}_A \otimes \mathbb{1}_B \Big[ \sin^2(gt)  \sum_{n,m,n',m'} p_{n,m,n',m'} ( \ket{m,n}\bra{m',n'} - \ket{n,m}\bra{n',m'}) \nonumber \\
    & + i \sin(gt) \cos(gt) \sum_{n,m,n',m'} p_{n,m,n',m'} (\ket{n,m}\bra{m',n'} - \ket{m,n}\bra{n',m'}) \Big] \Bigg),
\end{align}
where in the last equality we just used $S\ket{a,b} = \ket{b,a}$ and $\bra{a,b}S = \bra{b,a}$. Continuing with the computation, applying $\sum_k \omega_k \ket{k}_A \bra{k}_A \otimes \mathbb{1}_B$ in the operators at its right and summing over the Kronecker-deltas, we have
\begin{align}
    \langle \mathcal{Q}_A \rangle & = \tr \Bigg( \sin^2(gt)  \sum_{n,m,n',m'} p_{n,m,n',m'} ( \omega_m \ket{m,n}\bra{m',n'} - \omega_n \ket{n,m}\bra{n',m'}) \nonumber \\
    & + i \sin(gt) \cos(gt) \sum_{n,m,n',m'} p_{n,m,n',m'} (\omega_n \ket{n,m}\bra{m',n'} - \omega_m \ket{m,n}\bra{n',m'}) \Bigg) \nonumber \\
    & = \sin^2(gt) \sum_{n,m} p_{n,m,n,m} ( \omega_m - \omega_n) + i \sin(gt) \cos(gt) \sum_{n,m} p_{n,m,m,n} (\omega_n  - \omega_m), \nonumber
\end{align}
where in the last equality we used that $\tr(\ket{a,b} \bra{c,d}) = \delta_{a,c}\delta_{b,d}$. The equation above is the same as Eq. \eqref{heat_qudit_PSWAP}, this becomes clear if we use the marginals $p_{n,n}^A = \sum_m p_{n,m,n,m}$, $p_{m,m}^B = \sum_n p_{n,m,n,m}$ in the first term of r.h.s. of the equation. For the second term, we use that $\rho$ is hermitian, therefore $p_{n,m,m,n}^* = p_{m,n,n,m}$, and this implies that 
\begin{align}
    \sum_{n,m} p_{n,m,m,n} (\omega_n  - \omega_m) = \sum_{n,m} \omega_n (p_{n,m,m,n}-p_{n,m,m,n}^*) = 2i \sum_{n,m} \omega_n \text{Im}(p_{n,m,m,n}),
\end{align}
where in the first equality we used that $\sum_{n,m} p_{n,m,m,n} \omega_m = \sum_{m,n} p_{m,n,n,m} \omega_n$.
}

\section{Relevant parameters for the two qutrits and partial SWAP interactions case}\label{Appendix_C}

In Section \ref{Subsection_PS_qutrits}, we employed local thermal states. This implies that the diagonal terms $p_i$ ($i=0,1, \dots ,8$) correspond to those in the thermal density matrix $
\rho_{3 \otimes 3}$,
where $
\rho_{J} = \frac{e^{-\beta_{J}H_A}}{Z_J},$ with $H_A=H_B$ given by Eq. (\ref{qutrit hamiltonian}). Consequently, we can explicitly express $p_i$ as follows:
\begin{equation*}
p_{0} = \frac{e^{-\left(\beta_A+\beta_B\right)\omega_0}}{Z_AZ_B}, \quad p_{1} = \frac{e^{\left(-\beta_A\omega_0-\beta_B\omega_1\right)}}{Z_AZ_B}, \quad p_{2} = \frac{e^{\left(-\beta_A\omega_0-\beta_B\omega_2\right)}}{Z_AZ_B}, \quad p_{3} = \frac{e^{\left(-\beta_B\omega_0-\beta_A\omega_1\right)}}{Z_AZ_B},
\end{equation*}

\begin{equation*}
p_{4} = \frac{e^{-\left(\beta_A+\beta_B\right)\omega_1}}{Z_AZ_B}, \quad p_{5} = \frac{e^{\left(-\beta_A\omega_1-\beta_B\omega_2\right)}}{Z_AZ_B}, \quad p_{6} = \frac{e^{\left(-\beta_B\omega_0-\beta_A\omega_2\right)}}{Z_AZ_B}, \quad p_{7} = \frac{e^{\left(-\beta_B\omega_1-\beta_A\omega_2\right)}}{Z_AZ_B},
\end{equation*}
and
\begin{equation*}
p_{8}= \frac{e^{-\left(\beta_A+\beta_B\right)\omega_2}}{Z_AZ_B}.
\end{equation*}
Utilizing these results in Eq.~\eqref{heat for qutrit}, we obtain that
\begin{equation}
\langle \mathcal{Q}_A \rangle = \zeta \sin^2(gt) + \xi \sin(gt) \cos(gt), \label{qutrit_heat_thermal}
\end{equation}
where $\zeta$ is given by
\begin{equation*}
\zeta = -\omega_0 \left(p_3 + p_6 - p_1 - p_2\right)+ \omega_1 \left(p_1 - p_3 - p_5 + p_7\right) + \omega_2 \left(p_2 + p_5 - p_6 - p_7\right),
\end{equation*}
indicating that $\zeta$ is solely a function of the energies $\omega_i$ and temperatures $\beta_J$. In contrast,
\begin{equation*}
\xi = \eta_{31}\left(\omega_2 - \omega_1\right) \sin(\theta_{31})
+\eta_{62} \left(\omega_3 - \omega_1\right) \sin(\theta_{62}) + \eta_{75} \left(\omega_3 - \omega_2\right) \sin(\theta_{75}),
\end{equation*}
encodes the system's correlations, represented by $\eta_{ij}$ and phases $\theta_{ij}$. Notably, if we select three states separated by the same quantum number, such that $\omega_1 - \omega_0 = \omega_2 - \omega_1 \equiv \Delta\omega$, and set the phases $\theta_{ij}=\pi/2$, we achieve the simplest form
\begin{equation*}
\xi = \Delta\omega \left( \eta_{31} + 2\eta_{62} + \eta_{75} \right),
\end{equation*}
which indicates that the inversion of heat flux is strictly determined by negative values of the correlations $\eta_{ij}$.

{
We have also computed the first-order term, in $gt$, for the change of the mutual information (the higher-order terms are too large and give no additional information for our purpose) {\color{black}
\begin{equation}
   \Delta\mathcal{I}(A:B) =  2 g t (\beta_A-\beta_B) \xi + \mathcal{O}(g^2 t^2).
\end{equation}}
This term is, as in the two-qubits case, equal to $(\beta_A - \beta_B)$ times the first-order term of the heat average of Eq.~\eqref{qutrit_heat_thermal}. It comes from the heat contribution for the change of the mutual information of Eq.~\eqref{mutual information change} and is the leading order term causing the initial negativity in the change of the mutual information.
}

\end{document}